%% file: main.tex
\documentclass[a4paper,UKenglish,cleveref, autoref, thm-restate]{lipics-v2021}

\usepackage{microtype}
\usepackage[ruled,vlined]{algorithm2e}
\usepackage{wrapfig}

\title{R-enum: Enumeration of Characteristic Substrings in BWT-runs Bounded Space}
\titlerunning{Enumeration of Characteristic Substrings in BWT-runs Bounded Space}

\author{Takaaki Nishimoto}{RIKEN Center for Advanced Intelligence Project, Japan}{takaaki.nishimoto@riken.jp}{}{}
\author{Yasuo Tabei}{RIKEN Center for Advanced Intelligence Project, Japan}{yasuo.tabei@riken.jp}{}{}
\authorrunning{T. Nishimoto and Y. Tabei}
\Copyright{T. Nishimoto and Y. Tabei}


\category{}
\relatedversion{}
\supplement{}
\funding{}
\acknowledgements{}
\nolinenumbers 

\EventEditors{John Q. Open and Joan R. Access}
\EventNoEds{2}
\EventLongTitle{42nd Conference on Very Important Topics (CVIT 2016)}
\EventShortTitle{CVIT 2016}
\EventAcronym{CVIT}
\EventYear{2016}
\EventDate{December 24--27, 2016}
\EventLocation{Little Whinging, United Kingdom}
\EventLogo{}
\SeriesVolume{42}
\ArticleNo{23}

\begin{document}

\maketitle

\input{abstract}
\input{macros}

\input{1_intro}
\input{2_preliminaries}

\input{3_algorithm}

\input{4_applications.tex}

\input{5_option.tex}
\input{6_experiments}



\bibliography{ref}

\clearpage

\input{appendix}

\end{document}

%% file: abstract.tex
\begin{abstract}
Enumerating characteristic substrings (e.g., maximal repeats, minimal unique substrings, and minimal absent words) 
in a given string has been an important research topic because 
there are a wide variety of applications in various areas such as string processing and computational biology.
Although several enumeration algorithms for characteristic substrings have been proposed, 
they are not space-efficient in that their space-usage is proportional to the length of an input string. 
Recently, the run-length encoded Burrows-Wheeler transform (RLBWT) has attracted increased attention in string processing, 
and various algorithms for the RLBWT have been developed.
Developing enumeration algorithms for characteristic substrings with the RLBWT, 
however, remains a challenge.
In this paper, we present \emph{r-enum (RLBWT-based enumeration)}, the first enumeration algorithm for characteristic substrings based on RLBWT.
R-enum runs in $O(n \log \log (n/r))$ time and with $O(r \log n)$ bits of 
working space for string length $n$ and number $r$ of runs in RLBWT, where $r$ is expected to be significantly 
smaller than $n$ for highly repetitive strings (i.e., strings with many repetitions). 
Experiments using a benchmark dataset of highly repetitive strings show that the results of r-enum are more space-efficient than the previous results. 
In addition, we demonstrate the applicability of r-enum to a huge string by performing experiments on a 300-gigabyte string of 100 human genomes.
\end{abstract}

%% file: macros.tex
\newcommand{\Occ}{\mathit{Occ}}

\newcommand{\floor}[1]{\left \lfloor #1 \right \rfloor}
\newcommand{\ceil}[1]{\left \lceil #1 \right \rceil}

\newcommand{\argmax}{\mathop{\rm arg~max}\limits}
\newcommand{\argmin}{\mathop{\rm arg~min}\limits}
\newcommand{\polylog}{\mathop{\rm polylog}\limits}

\newcommand{\LCP}{\mathsf{LCP}}
\newcommand{\SA}{\mathsf{SA}}

\newcommand{\BWT}{\mathsf{BWT}}
\newcommand{\LF}{\mathsf{LF}}
\newcommand{\FL}{\mathsf{FL}}

\newcommand{\rank}{\mathsf{rank}}
\newcommand{\children}{\mathsf{children}}
\newcommand{\RD}{\mathsf{RD}}
\newcommand{\interval}{\mathsf{interval}}

\newcommand{\repr}{\mathsf{repr}}
\newcommand{\erepr}{\mathsf{eRepr}}

\newcommand{\Weiner}{\mathsf{WLink}}
\newcommand{\rightarr}{\mathsf{Right}}
\newcommand{\substr}{\mathsf{substr}}
\newcommand{\stinterval}{\mathsf{interval}}
\newcommand{\eRD}{\mathsf{eRD}}
\newcommand{\outputrepr}{\mathsf{output}}
\newcommand{\startset}{\mathcal{S}_{\mathsf{start}}}





%% file: 1_intro.tex
\section{Introduction}\label{sec:intro}
Enumerating characteristic substrings (e.g., maximal repeats, minimal unique substrings and minimal absent words) in 
a given string has been an important research topic because 
there are a wide variety of applications in various areas such as string processing and computational biology. 
The usefulness of the enumeration of maximal repeats has been demonstrated in lossless data compression~\cite{DBLP:conf/dcc/FuruyaTNIBK19},
bioinformatics~\cite{DBLP:journals/bioinformatics/BecherDH09,DBLP:books/cu/Gusfield1997} and 
and string classification with machine learning models~\cite{DBLP:conf/sdm/OkanoharaT09, DBLP:conf/iceis/MasadaTSO11}.
The enumeration of minimal unique substrings and minimal absent words has shown practical benefits in bioinformatics~\cite{DBLP:journals/bmcbi/HauboldPMW05,DBLP:journals/algorithms/AbedinKT20,DBLP:journals/iandc/Charalampopoulos18,DBLP:conf/latin/CrochemoreFMP16} and 
data compression~\cite{892711,DBLP:conf/sccc/CrochemoreN02}.
There is therefore a strong need to develop scalable algorithms for enumerating characteristic substrings in a huge string.

The \emph{Burrows-Wheeler transform (BWT)}~\cite{burrows1994block} is for permutation-based lossless data compression of a string, 
and many enumeration algorithms for characteristic substrings leveraging BWT have been proposed. 
Okanohara and Tsujii~\cite{DBLP:conf/sdm/OkanoharaT09} proposed an enumeration algorithm for maximal repeats that uses BWT and an enhanced suffix array~\cite{DBLP:journals/jda/AbouelhodaKO04}. 
Since their algorithm takes linear time to the length $n$ of a string and $O(n\log n)$ working space, 
applying it to a huge string is computationally demanding.
Beller et al.~\cite{DBLP:conf/spire/BellerBO12} proposed an enumeration algorithm for maximal repeats 
that uses a \emph{range distance query} on the BWT of a string in $O(n\log{\sigma})$ time and $n \log \sigma + o(n \log \sigma) + O(n)$ bits 
of working space for alphabet size $\sigma$ of a string.
Since the working space of their algorithm is linearly proportional to the length $n$ of a string, a large amount of space is expected to be consumed for huge strings.
Along the same line of research, Belazzougui et al.~\cite{DBLP:conf/cpm/BelazzouguiCGPR15} proposed an algorithm for enumerating characteristic substrings 
in $O(n\log{\sigma})$ time and $n \log \sigma + o(n \log \sigma) + O(\sigma^{2} \log^{2} n)$ bits of working space, 
which resulted in the space usage being linearly proportional to the string length.
Thus, developing a more space-efficient enumeration algorithm for the characteristic substrings of a string remains a challenging issue.

\emph{Run-length BWT~(RLBWT)} is a recent, popular lossless data compression, and it is defined as a run-length compressed BWT for strings.
Thus, the compression performance of RLBWT has been shown to be high, especially for \emph{highly repetitive strings} (i.e., strings with many repetitions) 
such as genomes, version-controlled documents, and source code repositories.
Kempa and Kociumaka~\cite{DBLP:journals/corr/abs-1910-10631} showed an upper bound on the size of the RLBWT by using a measure of repetitiveness.
Although several compressed data structures and string processing algorithms that use RLBWT have also been proposed~(e.g., \cite{DBLP:conf/cpm/BelazzouguiCGPR15,10.1145/3375890,DBLP:journals/jda/OhnoSTIS18,DBLP:conf/cpm/BannaiGI18,DBLP:conf/soda/Kempa19}), 
no previous algorithms for enumerating characteristic substrings based on RLBWT have been proposed. 
Such enumeration algorithms are expected to be much more space-efficient than existing algorithms for highly repetitive strings. 

\emph{Contribution.}
We present the first enumeration algorithm for characteristic substrings based on RLBWT, which we call \emph{r-enum (RLBWT-based enumeration)}. 
Following the idea of the previous works~\cite{DBLP:conf/spire/BellerBO12,DBLP:conf/spire/BelazzouguiC15}, 
r-enum performs an enumeration by simulating traversals of a \emph{Weiner-link tree}~(e.g., \cite{DBLP:journals/talg/BelazzouguiCKM20}), 
which is a trie, each node of which represents a right-maximal repeat in a string $T$ of length $n$.
Each characteristic substring in $T$ corresponds to a node in the Weiner-link tree of $T$. 
This is made possible in $O(n \log \log_{w} (n/r) + occ)$ time and $O(r \log n)$ bits of working space 
for $r$ number of runs in RLBWT, machine word size $w = \Theta(\log n)$, and $occ$ number of characteristic substrings.
For a highly repetitive string such that $r = o(n \log \sigma / \log n)$ holds, 
r-enum is more space-efficient than the best previous algorithms taking $O(nd)$ time and 
$|RD| + O(n)$ bits of space, where $|RD|$ is the size of a data structure supporting range distinct queries 
and computing the LF function in $O(d)$ time; 
a pair $(|RD|, d)$ can be chosen as $(|RD|, d) = (n \log \sigma + o(n \log \sigma), O(\log \sigma))$~\cite{DBLP:journals/is/ClaudeNP15} or $(|RD|, d) = (O(n \log \sigma), O(1))$~\cite[Lemmas 3.5 and 3.17]{DBLP:journals/talg/BelazzouguiCKM20}. 
Table~\ref{table:result} summarizes the running time and working space of state-of-the-art algorithms 
including those by Okanohara and Tsujii (OT method)~\cite{DBLP:conf/sdm/OkanoharaT09}, 
Beller et al. (BBO method)~\cite{DBLP:conf/spire/BellerBO12}, and 
Belazzougui and Cunial (BC method)~\cite{DBLP:conf/spire/BelazzouguiC15} in comparison with our r-enum. 

Experiments using a benchmark dataset of highly repetitive strings show that r-enum is more space-efficient than the previous algorithms. 
In addition, we demonstrate the applicability of r-enum to a huge string by performing experiments on a 300-gigabyte string of 100 human genomes, 
which has not been shown in the previous work so far.

The outline of this paper is as follows.
Section~\ref{sec:preliminary} introduces several basic notions, including the Weiner-link tree. 
In Section~\ref{sec:traverse}, we present a traversal algorithm for the Weiner-link tree of $T$ in $O(r \log n)$ bits.
Section~\ref{sec:applications} presents r-enum for finding
 the corresponding nodes to maximal repeats, minimal unique substrings, and minimal absent words.
In Section~\ref{sec:option}, we slightly modify r-enum 
such that it outputs each characteristic substring and its occurrences in $T$ instead of the corresponding node to the characteristic substring. 
Section~\ref{sec:exp} shows the performance of our method on benchmark datasets of highly repetitive strings.

\renewcommand{\arraystretch}{0.7}
\begin{table}[t]
    \vspace{-0.5cm}
    \caption{
    Summary of running time and working space of enumeration algorithms for 
    (i) maximal repeats, (ii) minimal unique substrings, and (iii) minimal absent words for each method. 
    Last column represents main data structure used in each algorithm. 
    Input of each algorithm is string $T$ of length $n$ or BWT of $T$, 
    and each outputted characteristic substring is represented by pointer with $O(\log n)$ bits. 
    We exclude inputs and outputs from working space. 
    In addition, $\sigma$ is alphabet size of $T$, 
    $w = \Theta(\log n)$ is machine word size, $r$ is the number of runs in RLBWT of $T$, 
    and $occ = O(n \sigma)$~\cite{DBLP:journals/ipl/CrochemoreMR98} is the number of minimal absent words for $T$. 
    RD means data structure (i) supporting range distinct queries 
    in $O(d)$ time per output element and (ii) computing LF function in $O(d)$ time; 
    $|RD|$ is its size. 
    We can choose $(|RD|, d) = (n \log \sigma + o(n \log \sigma), O(\log \sigma))$~\cite{DBLP:journals/is/ClaudeNP15} or $(|RD|, d) = (O(n \log \sigma), O(1))$~\cite[Lemmas 3.5 and 3.17]{DBLP:journals/talg/BelazzouguiCKM20}.
    }
    \label{table:result} 
    \center{	    
    \begin{tabular}{r||c|c|c}
(i) Maximal repeats & Running time & Working space~(bits) & Data structures \\ \hline
OT method~\cite{DBLP:conf/sdm/OkanoharaT09} & $O(n)$ & $O(n \log n)$ & Enhanced suffix array \\ \hline
\cite[Theorem 7.8]{DBLP:journals/talg/BelazzouguiCKM20} & $O(n)$ & $O(n \log \sigma)$ & BWT and RD \\ \hline
BBO method~\cite{DBLP:conf/spire/BellerBO12} & $O(n d)$ & $|RD| + O(n)$ & BWT and RD \\ \hline
BC method~\cite{DBLP:conf/spire/BelazzouguiC15} & $O(n d)$ & $|RD| + O(\sigma^{2} \log^{2} n)$ & BWT and RD \\ \hline \hline

r-enum (this study) & $O(n \log \log_{w} (n/r))$ & $O(r \log n)$ & RLBWT and RD 

    \end{tabular} 
    \smallskip    
    \smallskip    
    
    \begin{tabular}{r||c|c|c}
(ii) Minimal unique substrings & Running time & Working space~(bits)  & Data structures \\ \hline
BC method~\cite{DBLP:conf/spire/BelazzouguiC15} & $O(nd)$ & $|RD| + O(\sigma^{2} \log^{2} n)$  & BWT and RD \\ \hline \hline
r-enum (this study) & $O(n \log \log_{w} (n/r))$ & $O(r \log n)$ & RLBWT and RD \\ 
    \end{tabular} 
    \smallskip    
    \smallskip    
    
    \begin{tabular}{r||c|c|c}
(iii) Minimal absent words & Running time & Working space~(bits) & Data structures \\ \hline
\cite{DBLP:journals/bmcbi/BartonHMP14} & $O(n + occ)$ & $O(n \log n)$ & Suffix array \\ \hline
\cite[Theorem 7.12]{DBLP:journals/talg/BelazzouguiCKM20} & $O(n + occ)$ & $O(n \log \sigma)$ & BWT and RD \\ \hline
BC method~\cite{DBLP:conf/spire/BelazzouguiC15} & $O(nd + occ)$ & $|RD| + O(\sigma^{2} \log^{2} n)$ & BWT and RD \\ \hline \hline
r-enum (this study) & $O(n \log \log_{w} (n/r) + occ)$ & $O(r \log n)$ & RLBWT and RD \\ 
    \end{tabular} 

    }
\end{table}

%% file: 2_preliminaries.tex
\section{Preliminaries} \label{sec:preliminary}
Let $\Sigma = \{ 1, 2, \ldots, \sigma \}$ be an ordered alphabet, 
$T$ be a string of length $n$ over $\Sigma$, and $|T|$ be the length of $T$. 
Let $T[i]$ be the $i$-th character of $T$~(i.e., $T = T[1], T[2], \ldots, T[n]$), and $T[i..j]$ be the substring of $T$ that begins at position $i$ and ends at position $j$. 
For two strings $T$ and $P$, $T \prec P$ means that $T$ is lexicographically smaller than $P$. 
$\Occ(T, P)$ denotes all the occurrence positions of $P$ in $T$, i.e., $\Occ(T, P) = \{i \mid i \in [1, n-|P|+1] \mbox{ s.t. } P = T[i..(i+|P|-1)] \}$. 
We assume that 
(i) the last character of $T$ is a special character $\$$ not occurring in substring $T[1..n-1]$, 
(ii) $|T| \geq 2$, 
and (iii) every character in $\Sigma$ occurs at least once in $T$. 
For two integers $b$ and $e$~($b \leq e$), \emph{interval} $[b, e]$ represents the set $\{b, b+1, \ldots, e \}$. 
Let $\substr(T)$ denote the set of all the distinct substrings of $T$~(i.e., $\substr(T) = \{ T[i..j] \mid 1 \leq i \leq j \leq n \}$). 

In this paper, characteristic substrings of a string consist of \emph{maximal repeats}, \emph{minimal unique substrings}, and \emph{minimal absent words}.
A maximal repeat in $T$ is defined as a substring $P$ satisfying two conditions:
(i) it occurs at least twice in the string~(i.e., $|\Occ(T, P)| \geq 2$), 
and (ii) either of the left or right extended substrings of it 
occurs fewer times than it~(i.e., $|\Occ(T, cP)|, |\Occ(T, Pc)| < |\Occ(T, P)|$ for $c \in \Sigma$). 
A minimal unique substring is defined as substring $P$ satisfying two conditions:
(i) it occurs just once~(i.e., $|\Occ(T, P)| = 1$), 
and (ii) all the substrings of it occur at least twice in the string~(i.e., $|\Occ(T, P[2..|P|])|, |\Occ(T, P[1..|P|-1])| \geq 2$). 
A minimal absent word is defined as string $P$ satisfying two conditions:
(i) it does not occur in a string~(i.e., $|\Occ(T, P)| = 0$), and (ii) all the proper substrings of it occur in the string~(i.e., 
$|\Occ(T, P[2..|P|])|, |\Occ(T, P[1..|P|-1])| \geq 1$).
For convenience, a minimal absent word is sometimes called a substring, although the string is not a substring of $T$. 

Our computation model is a unit-cost word RAM with a machine word size of $w = \Theta(\log_2 n)$ bits. We evaluate the space complexity in terms of the number of machine words. A bitwise evaluation of the space complexity can be obtained with a multiplicative factor of $\log_2 n$. We assume the base-2 logarithm throughout this paper when the base is not indicated. 

\subsection{Rank and range distinct queries}\label{sec:preliminary_queries}
Let $S \subseteq \{1, 2, \ldots, n \}$ be a set of $d$ integers.  
A \emph{rank} query $\rank(S, i)$ on $S$ returns the number of elements no more than $i$ in $S$,
i.e., $\rank(S, i) = |\{ j \mid j \in S \mbox{ s.t. } j \leq i \}|$. 
$R_{\rank}(S)$ is a \emph{rank data structure} solving 
a rank query on $S$ in $O(\log\log_{w} (n/d))$ time and with $O(dw)$ bits of space~\cite{DBLP:conf/esa/BelazzouguiN12}. 

A range distinct query, $\RD(T, b, e)$, on a string $T$ returns a set of 3-tuples $(c, p_{c}, q_{c})$ that consists of 
(i) a distinct character $c$ in $T[b..e]$, 
(ii) the first occurrence $p_{c}$ of the character $c$ for a given interval $[b..e]$ in $T$, 
and (iii) the last occurrence $q_{c}$ of the character $c$ for $[b..e]$ in $T$.
Formally, let $\Sigma(T[b..e])$ be a set of distinct characters in $T[b..e]$, i.e., $\Sigma(T[b..e]) = \{ T[i] \mid i \in [b,e] \}$.
Then, $\RD(T, b, e) = \{ (c, p_{c}, q_{c}) \mid c \in \Sigma(T[b..e]) \}$, 
where $p_{c} = \min (\Occ(T, c) \cap [b,e])$, and $q_{c} = \max (\Occ(T, c) \cap [b,e])$.
$R_{\RD}(T)$ is a \emph{range distinct data structure} solving 
a range distinct query on $T$ in $O(|\RD(T, b, e)| + 1)$ time and with $O(n \log \sigma)$ bits of space~\cite{DBLP:journals/jda/BelazzouguiNV13}.

\subsection{Suffix and longest common prefix arrays}
\begin{wrapfigure}{r}{50mm}
\begin{center}
	\includegraphics[width=0.4\textwidth]{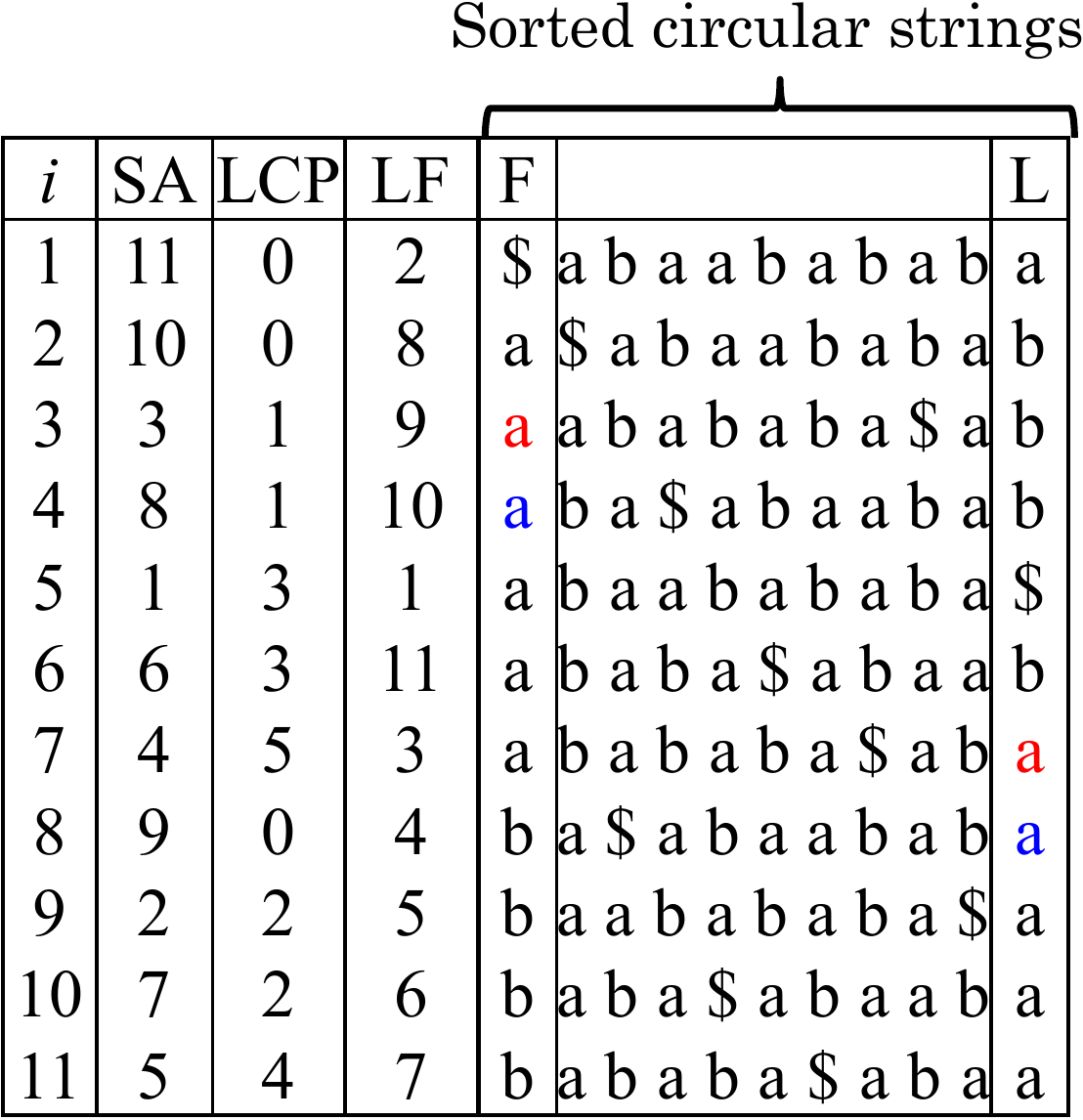}
	 \caption{Suffix array, LCP array, LF function, BWT, and circular strings for $T=abaabababa\$$.}
     \label{fig:main_figure}
\end{center}
\end{wrapfigure}

The \emph{suffix array}~\cite{DBLP:journals/siamcomp/ManberM93} $\SA$ of string $T$ is an integer array of size $n$ 
such that $\SA[i]$ stores the starting position of the $i$-th suffix of $T$ in lexicographical order.
Formally, $\SA$ is a permutation of $\{1, 2, \ldots,n \}$ such that $T[\SA[1]..n] \prec \cdots \prec T[\SA[n]..n]$. 
The \emph{longest common prefix array}~(LCP array) $\LCP$ of $T$ is 
an integer array of size $n$ such that $\LCP[1] = 0$ and $\LCP[i]$ stores 
the length of the LCP of the two suffixes $T[\SA[i]..n]$ and $T[\SA[i-1]..n]$ for $i \in \{2, 3, \ldots, n \}$.
We call the values in the suffix and LCP arrays \emph{sa-values} and \emph{lcp-values}, respectively. 
Moreover, let $\LF$ be a function such that (i) $\SA[\LF(i)] = \SA[i] - 1$ 
for any integer $i \in (\{1, 2, \ldots,n \} \setminus \{ x \})$ and 
(ii) $\SA[\LF(x)] = n$, 
where $x$ is an integer such that $\SA[x] = 1$. 
Figure~\ref{fig:main_figure} depicts the suffix array, LCP array, and LF function for a string. 


\subsection{BWT and RLBWT}\label{sec:bwt}
The BWT~\cite{burrows1994block} of string $T$ is an array $L$ built by permuting $T$ as follows;
(i) all the $n$ rotations of $T$ are sorted in lexicographical order, 
and (ii) $L[i]$ for any $i \in \{1, 2, \ldots,n \}$ is the last character at the $i$-th rotation in sorted order.
Similarly, $F[i]$ for any $i \in \{1, 2, \ldots,n \}$ is the first character at the $i$-th rotation in sorted order. 
Formally, let $L[i] = T[\SA[\LF(i)]]$ and $F[i] = T[\SA[i]]$. 

Since $L[i]$ and $L[j]$ represent two characters $T[\SA[i]-1]$ and $T[\SA[j]-1]$ for two distinct integers $i, j \in \{1, 2, \ldots, n \}$, 
the following relation holds between $\LF(i)$ and $\LF(j)$; 
$\LF(i) < \LF(j)$ if and only if either of two conditions holds:
(i) $L[i] \prec L[j]$ or (ii) $L[i] = L[j]$ and $i < j$. 
Let $C$ be an array of size $\sigma$ such that $C[c]$ is the number of occurrences of characters lexicographically smaller than $c\in \Sigma$ in string $T$,  
i.e., $C[c] = |\{ i \mid i \in [1, n] \mbox{ s.t. } T[i] \prec c \}|$. 
The following equation holds by the above relation for the LF function on BWT $L$: 
$\LF(i) = C[c] + \Occ(L[1..i], L[i])$. 
We call the equation \emph{LF formula}.

The RLBWT of $T$ is the BWT encoded by run-length encoding. 
i.e., RLBWT is a partition of $L$ into $r$ substrings $L[\ell(1)..\ell(2)-1], L[\ell(2)..\ell(3)-1], \ldots, L[\ell(r)..\ell(r+1)-1]$ 
such that each substring $L[\ell(i)..\ell(i+1)-1]$ is a maximal repetition of the same character in $L$ called a \emph{run}. 
Formally, let $r = 1 + |\{ i \mid i \in \{ 2, 3, \ldots, n \} \mbox{ s.t. } L[i] \neq L[i-1] \}|$, 
$\ell(1) = 1$, $\ell(r+1) = n+1$, and $\ell(j) = \min \{ i \mid i \in \{ \ell(j-1)+1, \ell(j-1)+2, \ldots, n \} \mbox{ s.t. } L[i] \neq L[i-1] \}$ 
for $j \in \{ 2, 3, \ldots, r \}$. 
Let $\startset$ denote the set of the starting position of each run in $L$, 
i.e., $\startset = \{ \ell(1), \ell(2), \ldots, \ell(r) \}$. 
RLBWT is represented as $r$ pairs $(L[\ell(1)], \ell(1))$, $(L[\ell(2)], \ell(2))$, $\ldots$, $(L[\ell(r)], \ell(r))$ using $2rw$ bits. 
$r \leq \sigma \leq n$ holds since we assume that 
every character in $\Sigma$ occurs in $T$. 

$\LF(i) = \LF(i-1) + 1$ holds for an integer $i \in \{2, 3, \ldots, n \}$ by LF formula if 
$i$ is not the starting position of a run in $L$~(i.e., $i \not \in \startset $). 
Similarly, $\LCP[\LF(i)] = 1 + \LCP[i]$ holds if $i \not \in \startset$. 

Figure~\ref{fig:main_figure} depicts two arrays $L$ and $F$ for string $T=abaabababa\$$. 
Since the BWT of $T=abaabababa\$$ is $abbb\$baaaaa$, 
the RLBWT of the string $T$ is $(a, 1), (b, 2), (\$, 5)$, $(b, 6)$, and $(a, 7)$. 
The red and blue characters $a$ are adjacent in a run on the BWT $L$. 
Hence, $\LF(8) = \LF(7) + 1$ holds by LF formula. Similarly, $\LCP[\LF(8)] = \LCP[8] + 1 = 1$ holds. 

\subsection{Suffix tree}
\begin{figure}
\begin{center}
	\includegraphics[width=0.45\textwidth]{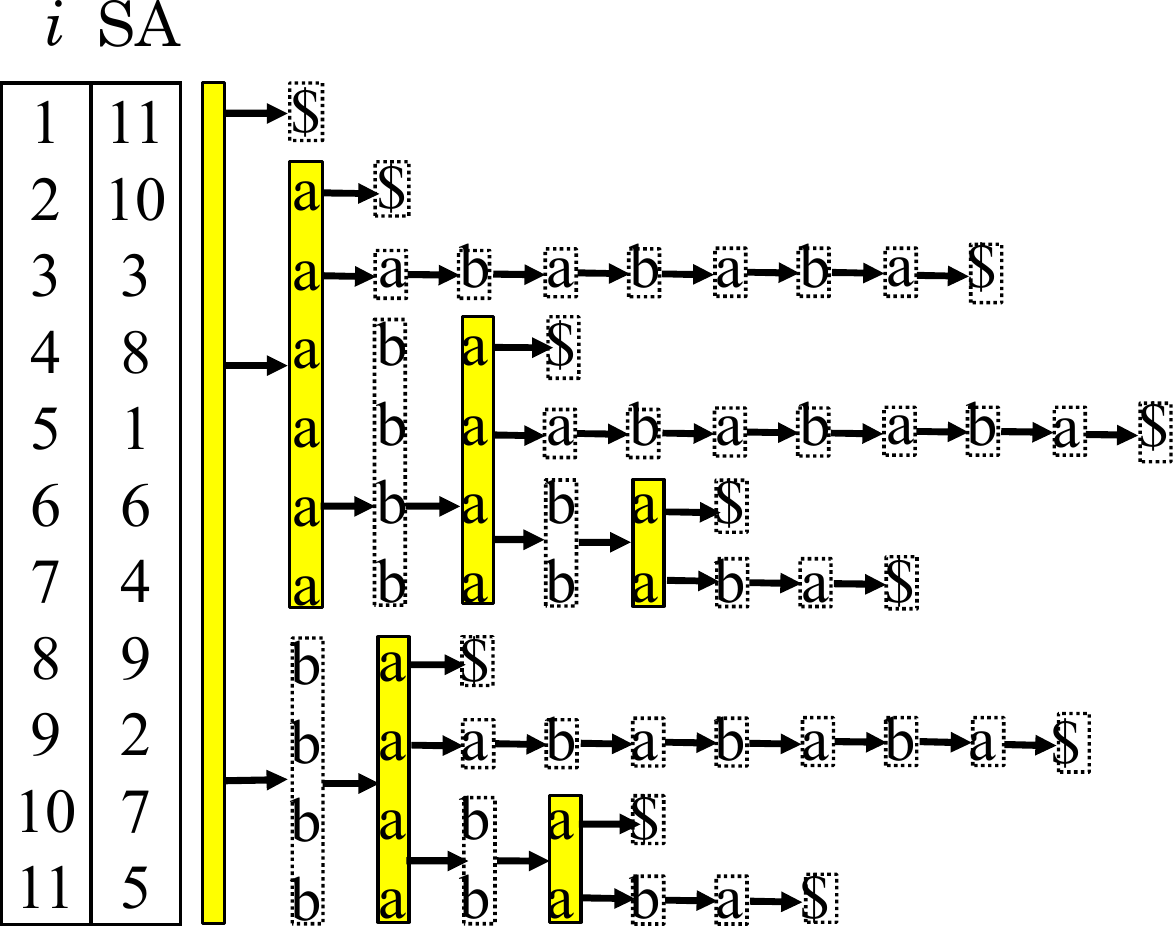}
	\includegraphics[width=0.45\textwidth]{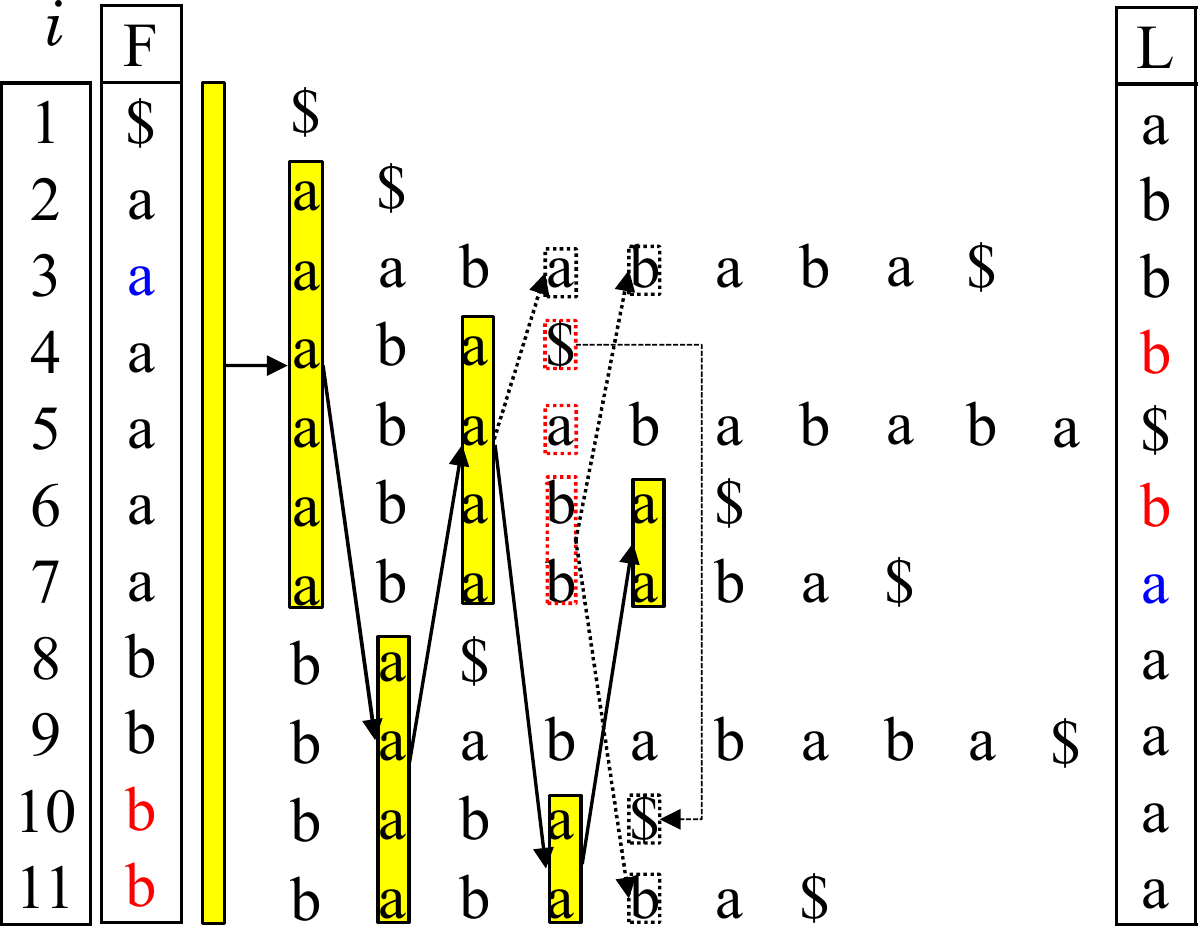}

	 \caption{
	 Left figure illustrates suffix tree of $T = abaabababa\$$ on sorted suffixes of $T$. 
	 We represent each node by its suffix-tree interval~(rectangles) and omit characters on edges. 
	 Yellow and white rectangles are explicit and implicit nodes, respectively. 
	 Right figure illustrates Weiner-link tree of $T$ and Weiner links on sorted suffixes of $T$. 
	 Tree consists of yellow rectangles and solid arrows. We omit characters on Weiner links. 
	 Solid and dotted arrows represent Weiner links pointing to explicit and implicit nodes, respectively. 
	 Red rectangles are children of node of $aba$ in suffix tree of $T$.
	 }
     \label{fig:suffix_tree}
\end{center}
\end{figure}

The \emph{suffix tree}~\cite{10.1007/978-3-642-82456-2_6} of $T$ is a trie storing all the suffixes of $T$. 
Each node $v$ represents the concatenated string on the path from the root of the tree to the node $v$. 
Let $u_{P}$ denote the node representing a substring $P$ of $T$.
The \emph{depth} of node $u_{P}$ is defined as the length of its string $P$, 
i.e., the depth of $u_{P}$ is $|P|$. 
Let $\children(P)$ denote the set of strings represented by the children of $u_{P}$.  
Formally, $\children(P) = \{ Pc \mid c \in \Sigma \mbox{ s.t. } Pc \in \substr(T) \}$. 
We call $u_{P}$ an \emph{explicit node} if it has at least two children; 
otherwise, we call $u_{P}$ an \emph{implicit node}. 
The root $u_{\varepsilon}$ of the suffix tree is explicit because $T$ contains at least two distinct characters. 
Let $\mathcal{L}_{d}$ be the set of substrings represented by all the explicit nodes with depth $d$. 
Formally, $\mathcal{L}_{d} = \{ P \mid P \in \substr(T) \mbox{ s.t. } |P| = d \mbox{ and } |\children(P)| \geq 2 \}$. 

A \emph{suffix-tree interval}~(a.k.a. \emph{suffix-array interval}) for a substring $P$ of $T$ 
is an interval $[b,e]$ on the suffix array of $T$ such that 
$\SA[b..e]$ represents all the occurrence positions of $P$ in string $T$; 
that is, for any integer $p \in \{1, 2, \ldots, n \}$, $T[p..p+|P|-1] = P$ holds if and only if $p \in \{ \SA[b], \SA[b+1], \ldots, \SA[e] \}$. 
The suffix-tree interval of the empty string $\varepsilon$ is defined as $[1, n]$. 
Let $\stinterval(P) = [b,e]$ denote the suffix-tree interval for $P$. 

\emph{Rich representation}~\cite{DBLP:journals/talg/BelazzouguiCKM20} $\repr(P)$ for $P$ is a 3-tuple $(\stinterval(P)$, $\{ (c_{1}, b_{1}, e_{1})$, $(c_{2}, b_{2}, e_{2})$, $\ldots$, $(c_{k}, b_{k}, e_{k}) \}$, $|P|)$. 
Here, $Pc_{1}, Pc_{2}, \ldots, Pc_{k}$ are strings represented by the children of node $u_{P}$, 
and $[b_{i}, e_{i}] = \stinterval(Pc_{i})$ for $i \in [1, k]$. 

Figure~\ref{fig:suffix_tree} illustrates the suffix tree of string $T = abaabababa\$$~(left figure). 
The three sets $\mathcal{L}_{0}$, $\mathcal{L}_{1}$, $\mathcal{L}_{2}$ are $\{ \varepsilon \}, \{ a \}, \{ ba \}$, respectively. 
The suffix-tree interval for substring $P = aba$ is $[4, 7]$, 
and $\children(P) = \{ aba\$, abaa, abab \}$. 
The rich representation for $P$ is $([4, 7], \{ (\$, 4, 4), (a, 5, 5), (b, 6, 7)  \}, 3)$. 

\subsection{Weiner links and Weiner-link tree}
\emph{Weiner links} are additional directed links on the suffix tree of $T$. 
Let $cP$ be a substring of $T$, 
where $c$ is a character, and $P$ is a string. 
Then, node $u_{cP}$ is the destination of a Weiner link with character $c$ starting at node $u_{P}$. 
Hence, every node $u_{P}$ must be the destination of exactly one Weiner link unless node $u_{P}$ is the root of the suffix tree~(i.e., $P = \varepsilon$). 
Node $u_{P}$ is always explicit if $u_{cP}$ is explicit because 
the explicit node $u_{cP}$ indicates that 
$T$ has two substrings $Pc_{1}$ and $Pc_{2}$, where $c_{1}$ and $c_{2}$ are distinct characters. 
Let $\Weiner(P)$ denote a set of strings such that 
each string represented by the destination node of a Weiner link starts at node $u_{P}$~(i.e., 
$\Weiner(P) = \{ cP \mid c \in \Sigma \mbox{ s.t. } cP \in \substr(T) \}$). 

A Weiner-link tree~(a.k.a \emph{suffix-link tree}) for $T$ is a graph such that 
(i) the nodes are all the explicit nodes in the suffix tree of $T$
and (ii) the edges are all the Weiner links among the explicit nodes. 
Since any explicit node is the destination of a Weiner link starting at another explicit node, 
the graph results in a tree. 
Each child of a node $u_{P}$ represents a string in $\Weiner(P)$ in the Weiner-link tree.

Let $p_{c}$ and $q_{c}$ be the first and last occurrences of a character $c$ on the suffix-tree interval for a substring $P$ in the BWT $L$ of $T$, respectively. 
Then, $[\LF(p_{c}), \LF(q_{c})]$ is equal to the suffix-tree interval of $cP$. 
Hence, we can compute the suffix-tree intervals for the destinations of all the Weiner links starting at node $u_{P}$ by using 
a range distinct query and LF function.
Formally, the following lemma holds. 
\begin{lemma}\label{lem:weiner}
$\{ \stinterval(cP) \mid cP \in \Weiner(P) \} = \{ [\LF(p_{c}), \LF(q_{c})] \mid (c, p_{c}, q_{c}) \in \RD(L, b, e) \\ \mbox{ s.t. } c \neq \$ \}$ holds for a substring $P$ of $T$, 
where $\stinterval(P) = [b,e]$. 
\end{lemma}
\begin{proof}
Since suffix $T[\SA[i]..n]$ has $P$ as a prefix for each $i \in \stinterval(P)$, 
two suffixes $T[\SA[\LF(p_{c})]..n]$ and $T[\SA[\LF(q_{c})]..n]$ are the lexicographically smallest and largest suffixes having $cP$ as a prefix. 
Hence, Lemma~\ref{lem:weiner} holds.
\end{proof}

Let $Q_{P}$ denote an array of size $\sigma$ for a substring $P$ of $T$ such that 
$Q_{P}[c]$ stores set $\{ (c', b, e) \mid cPc' \in \children(cP) \}$ for a character $c \in \Sigma$ if $cP$ is a substring of $T$; 
otherwise, $Q_{P}[c] = \emptyset$, 
where $[b,e] = \stinterval(cPc')$. 
The array $Q_{P}$ has three properties for any character $c \in \Sigma$: 
(i) $cP \in \Weiner(P)$ holds if and only if $|Q_{P}[c]| \geq 1$ holds, 
(ii) node $u_{cP}$ is explicit if and only if $|Q_{P}[c]| \geq 2$ holds, and 
(iii) $Q_{P}[c] = \bigcup_{Pc' \in \children(P)} \{ (c', b, e) \mid \hat{c}Pc' \in \Weiner(Pc') \mbox{ s.t. } \hat{c} = c \}$ holds.
In other words, each child of node $u_{cP}$ is the destination of a Weiner link starting at a child of node $u_{P}$ in the suffix tree.

Let $\mathcal{B}$ be a data structure supporting a range distinct query on BWT $L$ in $O((1+k) x)$ time and computing LF function in $O(x')$ time. 
Here, (i) $k$ is the number of elements output by the range distinct query, and (ii) $x$ and $x'$ are terms. 
We can compute the children of a node $u_{P}$ in the Weiner-link tree by using $Q_{P}$ and $\mathcal{B}$ without explicitly constructing the Weiner-link tree. 
Formally, the following lemma holds. 
\begin{lemma}[\cite{DBLP:conf/spire/BelazzouguiC15}]\label{lem:weiner_algorithm}
We can compute set $\mathcal{Y} = \{ \repr(cP) \mid cP \in \Weiner(P) \}$ 
by using (i) $\repr(P)$, (ii) data structure $\mathcal{B}$, and (iii) an empty array $X$ of size $\sigma$. 
After that, we can divide the set $\mathcal{Y}$ into 
two sets for explicit and implicit nodes. 
The computation time and working space are $O(h(x + x'))$ and $O( (\sigma + h')w)$ bits, respectively, 
where $h = \sum_{Pc' \in \children(P)} |\Weiner(Pc')|$, and 
$h' = \sum_{cP \in \Weiner(P) \cap \mathcal{L}_{|P|+1}} |\children(cP)|$. 
\end{lemma}
\begin{proof}
We compute the outputs with the following three steps. 
At the first step, 
we compute set $\{ \stinterval(cP) \mid cP \in \Weiner(P) \}$ and convert the empty array $X$ into $Q_{P}$ 
by using Lemma~\ref{lem:weiner} and the third property of $Q_{P}$.
At the second step, we output the rich representation $\repr(cP) = (\stinterval(cP), Q_{P}[c], |P|+1)$ for each $cP \in \Weiner(P)$ 
and divide the rich representations into two sets for explicit and implicit nodes by using the second property of $Q_{P}$. 
At the last step, 
we remove all the elements from $Q_{P}[c]$ for each $cP \in \Weiner(P)$ to recover $X$ from $Q_{P}$. 
We perform the three steps by using range distinct queries to the suffix-tree intervals for node $u_{P}$ and its children, 
and the intervals stored in $\repr(P)$. 
Hence, the running time is $O(h(x + x'))$ in total. 

Next, we analyze the working space. 
The rich representation for a node $u_{P'}$ takes $O(w)$ bits if $u_{P'}$ is implicit 
because it has at most one child. 
Otherwise, $\repr(P')$ takes $O(|\children(P')|w)$ bits. 
Hence, the working space is $O((\sigma + h')w)$ bits in total.  
\end{proof}

Figure~\ref{fig:suffix_tree} illustrates the Weiner-link tree of string $T = abaabababa\$$~(right figure). 
Since string $T$ contains two substrings $aaba$ and $baba$, 
the node of $aba$ has two Weiner links pointing to the nodes of $aaba$ and $baba$, i.e., $\Weiner(aba) = \{ aaba, baba \}$. 
The Weiner-link tree contains the node of $baba$ but not that of $aaba$ 
because the former and latter nodes are explicit and implicit, respectively. 

The three suffix-tree intervals for $aba$, $aaba$, and $baba$ are $[4, 7]$, $[3, 3]$, and $[10, 11]$, respectively. 
A range distinct query on $\stinterval(aba)$ in the BWT of $T$ returns the set $\{ (\$, 5, 5), (a, 7, 7), (b, 4, 6) \}$. 
Hence, $\stinterval(aaba) = [\LF(7)$, $\LF(7)]$ and $\stinterval(baba) = [\LF(4)$, $\LF(6)]$ hold 
by Lemma~\ref{lem:weiner}~(See also red and blue characters on two arrays $F$ and $L$ in Figure~\ref{fig:suffix_tree}). 
Figure~\ref{fig:suffix_tree} also illustrates the children of the node of $aba$ in the suffix tree and 
Weiner links starting from the children. 
In this example, $Q_{P}[\$] = \emptyset$, $Q_{P}[a] = \{ (b, 3, 3) \}$, and $Q_{P}[b] = \{ (\$, 10, 10), (b, 11, 11)  \}$, 
where $P = aba$. 
Hence, the node of $aaba$ has one child, and the node of $baba$ has two children in the suffix tree.

%% file: 3_algorithm.tex
\section{Traversing Weiner-link tree in \texorpdfstring{$O(rw)$}{O(rw)} bits of space}\label{sec:traverse}
In this section, we present a breadth-first traversal algorithm for the Weiner-link tree of $T$ in $O(rw)$ bits of working space. 
The traversal algorithm outputs $n$ sets $\{ \repr(P) \mid P \in \mathcal{L}_{0} \}$, $\{ \repr(P) \mid P \in \mathcal{L}_{1} \}$, $\ldots$, 
$\{ \repr(P) \mid P \in \mathcal{L}_{n-1} \}$ in left-to-right order. 
Here, each set $\{ \repr(P) \mid P \in \mathcal{L}_{t} \}$ represents the set of the rich representations for all the nodes with depth $t$ in the Weiner-link tree. 
Hence each node $u_{P}$ is represented as its rich representation $\repr(P)$. 

\subsection{Data structures}\label{sec:data_structures}
Our traversal algorithm uses six data structures: 
(i) the RLBWT of $T$~(i.e., $r$ pairs $(L[\ell(1)], \ell(1))$, $(L[\ell(2)], \ell(2))$, $\ldots$, $(L[\ell(r)], \ell(r))$, which are introduced in Section~\ref{sec:bwt}), 
(ii) the rich representation for the root of the Weiner-link tree of $T$~(i.e., $\repr(\varepsilon)$), 
(iii) $D_{\LF}$, 
(iv) $R_{\rank}(\startset)$, 
(v) $R_{\RD}(L')$, 
and (vi) an empty array $X$ of size $\sigma$.
$D_{\LF}$ is an array of size $r$ such that 
$D_{\LF}[i] = \LF(\ell(i))$ for $i \in \{ 1, 2, \ldots, r \}$. 
$R_{\rank}(\startset)$ is the rank data structure introduced in Section~\ref{sec:preliminary_queries}, 
and it is built on the set $\startset = \{ \ell(1), \ell(2), \ldots, \ell(r) \}$ introduced in Section~\ref{sec:bwt}. 
Similarly, $R_{\RD}(L')$ is the range distinct data structure introduced in Section~\ref{sec:preliminary_queries}, 
and it is built on string $L' = L[\ell(1)], L[\ell(2)], \ldots, L[\ell(r)]$. 
The six data structures require $O( (r + \sigma)w )$ bits in total. 
We construct the six data structures in $O(n \log \log_{w} (n/r))$ time and $O(rw)$ bits of working space by processing the RLBWT of $T$~(see Appendix A). 

We use the two data structures $D_{\LF}$ and $R_{\rank}(\startset)$ to compute LF function. 
Let $x$ be the index of a run containing the $i$-th character of BWT $L$~(i.e., $x = \rank(\startset, i)$) for an integer $i \in \{ 1, 2, \ldots, n \}$. 
$\LF(i) = \LF(\ell(x)) + |\Occ(L[\ell(x)..i], L[i])| - 1$ holds by LF formula,  
and $|\Occ(T[\ell(x)..i], L[i])| = i - \ell(x) + 1$ holds 
because 
$L[\ell(x)..i]$ consists of a repetition of the $i$-th character $L[i]$. 
Hence, $\LF(i) = D_{\LF}[x] + (i - \ell(x))$ holds, and 
we can compute LF function in $O(\log \log_{w} (n/r))$ time by using $D_{\LF}$ and $R_{\rank}(\startset)$.

We use the fifth data structure $R_{\RD}(L')$ to compute a range distinct query on the BWT $L$ of $T$. 
Let $b'$ and $e'$ be the indexes of the two runs on $L$ containing 
two characters $L[b]$ and $L[e]$, respectively~(i.e., $b' = \rank(\startset, b)$ and $e' = \rank(\startset, e)$), for an interval $[b,e] \subseteq \{1, 2, \ldots, n \}$. 
Then the following relation holds between two sets $\RD(L, b, e)$ and $\RD(L', b', e')$. 
\begin{lemma}\label{lem:range_query}
$\RD(L, b, e) = \{ (c, \max\{ \ell(p), b \}, \min \{ \ell(q+1)-1, e \}) \mid (c, p, q) \in \RD(L', b', e') \}$ holds. 
\end{lemma}
\begin{proof}
$L[b..e]$ consists of $e'-b'+1$ repetitions $L[b..\ell(b'+1) - 1]$, $L[\ell(b'+1)..\ell(b'+2)-1]$, $\ldots$, $L[\ell(e'-1)..\ell(e')-1]$, $L[\ell(e')..e]$.  
Hence, the following three statements hold: 
(i) The query $\RD(L', b', e')$ returns all the distinct characters in $L[b..e]$~(i.e., $\{ L[i] \mid i \in [b,e] \} = \{ L'[i] \mid i \in [b',e'] \}$). 
(ii) 
Let $p$ be the first occurrence of a character $c$ on $[b', e']$ in $L'$. 
Then, the first occurrence of a character $c$ on $[b,e]$ in $L$ is equal to $\ell(p)$ if $\ell(p) \geq b$; 
otherwise, the first occurrence is equal to $b$.  
(iii) Similarly, let $q$ be the last occurrence of the character $c$ on $[b', e']$ in $L'$. 
Then, the last occurrence of $c$ on $[b,e]$ in $L$ is equal to $\ell(q+1)-1$ if $\ell(q+1)-1 \leq e$; 
otherwise, the last occurrence is equal to $e$. 
We obtain Lemma~\ref{lem:range_query} with the three statements. 
\end{proof}
Lemma~\ref{lem:range_query} indicates that 
we can solve a range distinct query on $L$ by using two rank queries on $\startset$ and a range distinct query on $L'$. 
The range distinct query on $L$ takes $O(\log \log_{w} (n/r) + k)$ time, where $k$ is the number of output elements.

We use the empty array $X$ to compute rich representations by using Lemma~\ref{lem:weiner_algorithm}. 
The algorithm of Lemma~\ref{lem:weiner_algorithm} takes $O((k+1) \log \log_{w} (n/r))$ time by using the six data structures 
since we can compute LF function and solve the range distinct query on $L$ in $O(\log \log_{w} (n/r))$ and $O((k+1) \log \log_{w} (n/r))$ time, respectively.

\subsection{Algorithm}\label{sec:lcp-interval_algorithm}
The basic idea behind our breadth-first traversal algorithm is to traverse the Weiner-link tree without explicitly building the tree 
in order to reduce the working space. 
The algorithm computes nodes sequentially by using Lemma~\ref{lem:weiner_algorithm}. 
Recall that the algorithm of Lemma~\ref{lem:weiner_algorithm} returns the rich representations for all the children of a given explicit node $u_{P}$ 
in the Weiner-link tree~(i.e., it returns set $\{ \repr(cP) \mid cP \in \Weiner(P) \mbox{ s.t. } cP \in \mathcal{L}_{|P|+1} \}$). 
This fact indicates that 
the set of the rich representations for all the nodes with a depth $t \geq 1$ is equal to 
the union of the sets of rich representations obtained by applying Lemma~\ref{lem:weiner_algorithm} to all the nodes with depth $t-1$ in the Weiner-link tree, i.e., 
$\{ \repr(P) \mid P \in \mathcal{L}_{t}\} = \bigcup_{P \in \mathcal{L}_{t-1}} \{ \repr(cP) \mid cP \in \Weiner(P) \}$. 

Our algorithm consists of $(n-1)$ steps. 
At each $t$-th step, 
the algorithm (i) applies Lemma~\ref{lem:weiner_algorithm} to each representation in set $\{ \repr(P) \mid P \in \mathcal{L}_{t-1}\}$, 
(ii) outputs set $\{ \repr(P) \mid P \in \mathcal{L}_{t}\}$, 
and (iii) removes the previous set $\{ \repr(P) \mid P \in \mathcal{L}_{t-1}\}$ from working memory. 
The algorithm can traverse the whole Weiner-link tree in breadth-first order 
because we initially store the first set $\{ \repr(P) \mid P \in \mathcal{L}_{0} \}$ for the first step.

\subsection{Analysis}\label{sec:analysis}
\begin{figure}
\begin{center}
	\includegraphics[width=0.3\textwidth]{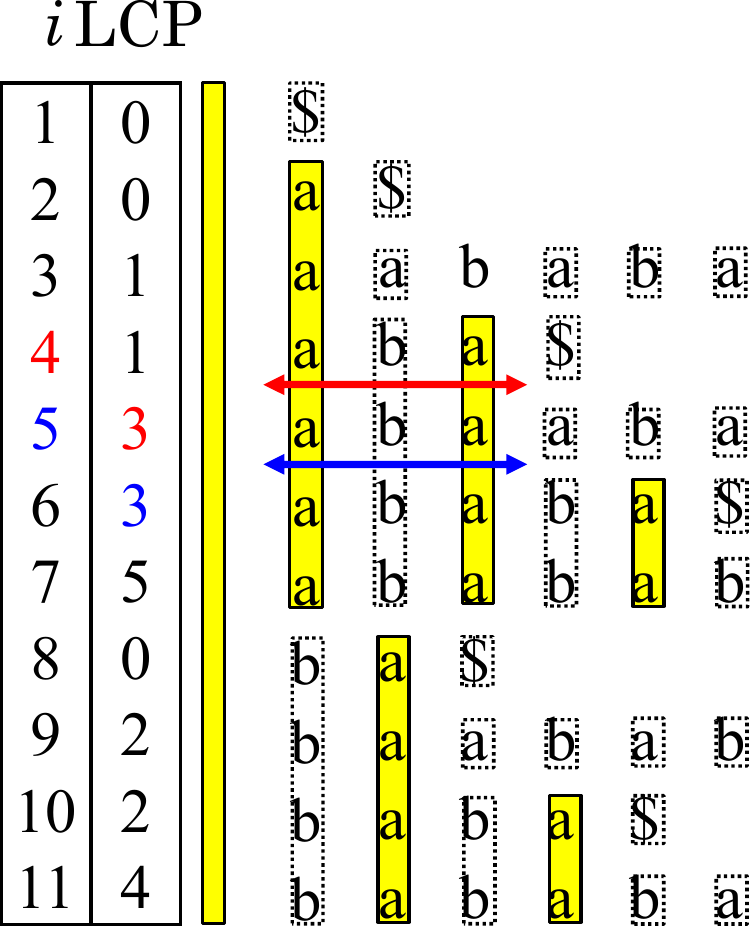}
	\includegraphics[width=0.3\textwidth]{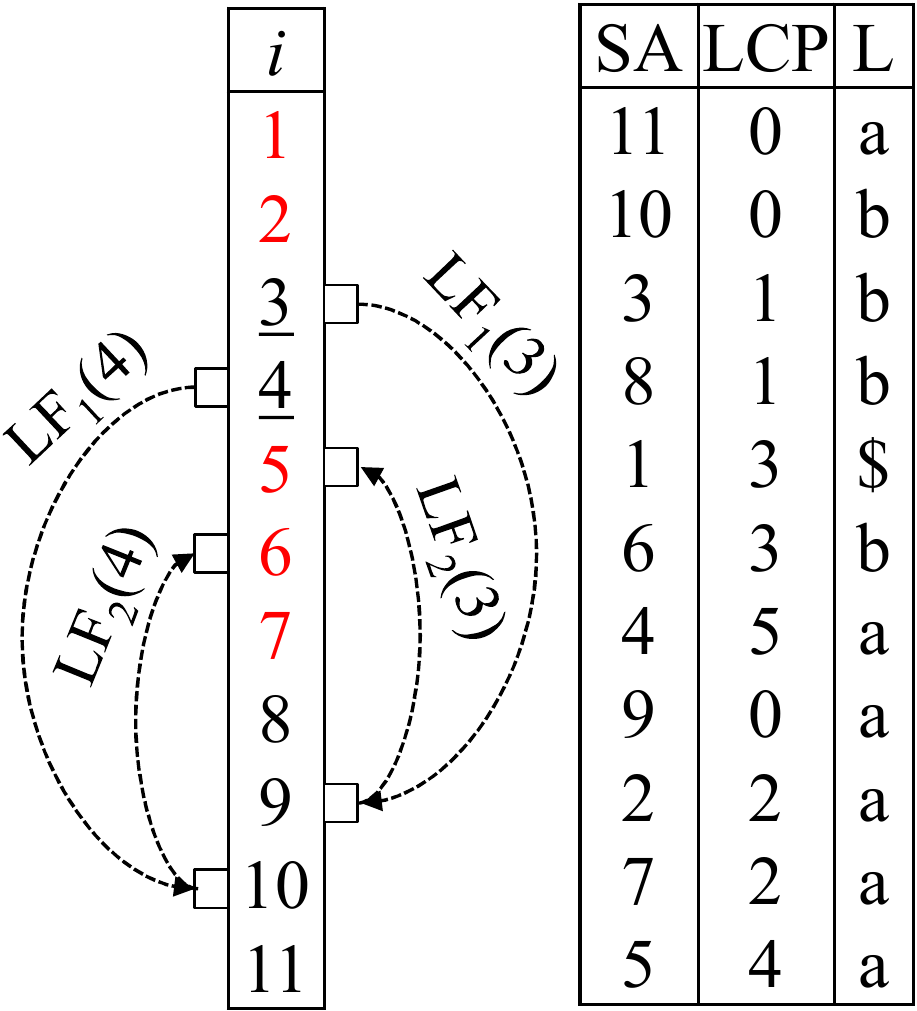}

	 \caption{
	 Left figure illustrates part of suffix tree of $T = abaabababa\$$ on sorted suffixes of $T$. 
	 Two colored integers in array named $i$ are positions in $\mathcal{K}'_{2} = \{ 4, 5 \}$. 
	 Similarly, two colored integers in LCP array correspond to positions in $\mathcal{P}_{2} = \{ 5, 6 \}$. 
	 Right figure illustrates array $i$, SA, LCP array, and BWT $L$ for $T = abaabababa\$$. 
	 Red integers in array $i$ are integers in set $\startset = \{ 1, 2, 5, 6, 7 \}$. 
	 Similarly, underlined integers are integers in set $\mathcal{P}_{1} = \{ 3, 4 \}$. 
	 }
     \label{fig:analysis}
\end{center}
\end{figure}

The traversal algorithm requires $O(H_{t-1} \log \log_{w} (n/r))$ computation time at the $t$-th step. 
Here, $H_{t}$ is the number of Weiner links starting from the children of the explicit nodes 
with depth $t$~(i.e., $H_{t} = \sum_{P \in \mathcal{L}_{t}} \sum_{Pc \in \children(P)} |\Weiner(Pc)|$). 
The running time is $O((\sum_{t = 0}^{n} H_{t}) \log \log_{w} (n/r))$ in total. 
The term $\sum_{t = 0}^{n} H_{t}$ represents the number of Weiner links starting from the children of all the explicit nodes in the suffix tree. 
Belazzougui et al. showed that $\sum_{t = 0}^{n} H_{t} = O(n)$ holds~\cite[Observation 1]{DBLP:journals/talg/BelazzouguiCKM20}. 
Hence our traversal algorithm runs in $O(n \log \log_{w} (n/r))$ time. 

Next, we analyze the working space of the traversal algorithm. 
Let $\mathcal{K}_{t}$ be the set of the children of 
all the explicit nodes with depth $t$ in the suffix tree of $T$~(i.e., $\mathcal{K}_{t} = \bigcup_{P \in \mathcal{L}_{t}} \children(P)$).  
Then, the algorithm requires $O( (r + \sigma + |\mathcal{K}_{t-1}| + |\mathcal{K}_{t}|)w)$ bits of working space 
while executing the $t$-th step. 
Hence our traversal algorithm requires $O( (r + \sigma + \max \{ |\mathcal{K}_{0}|, |\mathcal{K}_{1}|, \ldots, |\mathcal{K}_{n}| \} )w ))$ bits of working space 
while the algorithm runs. 

We introduce a set $\mathcal{K}'_{t}$ to analyze the term $|\mathcal{K}_{t}|$ for an integer $t$. 
The set $\mathcal{K}'_{t}$ consists of 
the children of all the explicit nodes with depth $t$ except for the last child of each explicit node, 
i.e., $\mathcal{K}'_{t} = \bigcup_{P \in \mathcal{L}_{t}} \{ Pc \mid Pc \in \children(P) \mbox{ s.t. } e' \neq e \}$, 
where $[b,e] = \stinterval(P)$, and $[b',e'] = \stinterval(Pc)$. 
$|\mathcal{K}_{t}| = |\mathcal{K}'_{t}| + |\mathcal{L}_{t}|$ holds 
because every explicit node has exactly one last child.  
$|\mathcal{L}_{t}| \leq |\mathcal{K}'_{t}|$ also holds  
because every explicit node has at least two children in the suffix tree of $T$. 
Hence, we obtain the inequality $|\mathcal{K}_{t}| \leq 2|\mathcal{K}'_{t}|$. 

We also introduce a set $\mathcal{P}_{t}$ for an integer $t$. 
The set $\mathcal{P}_{t}$ consists of positions with lcp-value $t$ on the LCP array of $T$ except for the first lcp-value $LCP[1]$, 
i.e., $\mathcal{P}_{t} = \{ i \mid i \in \{ 2, 3, \ldots, n \} \mbox{ s.t. } \LCP[i] = t \}$. 
Let $P$ be the longest common prefix of $T[\SA[i-1]..n]$ and $T[\SA[i]..n]$ for position $i \in \mathcal{P}_{t}$, 
i.e., $P = T[\SA[i]..\SA[i] + \LCP[i] - 1]$. 
Then node $u_{P}$ is explicit, and 
two nodes $u_{Pc}$ and $u_{Pc'}$ are adjacent children of $u_{P}$, 
where $c = T[\SA[i-1]+ \LCP[i]]$ and $c' = T[\SA[i]+ \LCP[i]]$. 
The two adjacent children $u_{Pc}$ and $u_{Pc'}$~($Pc \prec Pc'$) of an explicit node $u_{P}$ 
indicate that $\LCP[i] = |P|$ holds, where $i$ is the left boundary of $\stinterval(Pc')$. 
Hence, there exists a one-to-one correspondence between $\mathcal{P}_{t}$ and $\mathcal{K}'_{t}$, 
and $|\mathcal{P}_{t}| = |\mathcal{K}'_{t}|$ holds. 
We obtain the inequality $|\mathcal{K}_{t}| \leq 2|\mathcal{P}_{t}|$ by $|\mathcal{P}_{t}| = |\mathcal{K}'_{t}|$ and $|\mathcal{K}_{t}| \leq 2|\mathcal{K}'_{t}|$.

In Figure~\ref{fig:analysis}, the left figure represents a part of a suffix tree. 
In this example, $\mathcal{P}_{3} = \{ 5, 6 \}$, and $\mathcal{K}'_{3} = \{ 4, 5 \}$. 
Obviously, the two positions $5$ and $6$ in $\mathcal{P}_{3}$ correspond to the two positions $4$ and $5$ in $\mathcal{K}'_{3}$ 
with the adjacent children of explicit node $u_{aba}$, respectively.

Next, we introduce two functions $\LF_{x}$ and $\mathsf{dist}(i)$ to analyze $|\mathcal{P}_{t}|$.
The first function $\LF_{x}(i)$ returns the position obtained by recursively applying LF function to $i$ $x$ times, 
i.e., $\LF_{0}(i) = i$ and $\LF_{x}(i) = \LF_{x-1}(\LF(i))$ for $x \geq 1$. 
The second function $\mathsf{dist}(i) \geq 0$ returns the smallest integer such that $\LF_{\mathsf{dist}(i)}(i)$ is the starting position of a run on $L$~(i.e., 
$\LF_{\mathsf{dist}(i)}(i) \in \startset $). 
Formally, $\mathsf{dist}(i) = \min \{ x \mid x \geq 0 \mbox{ s.t. } \LF_{x}(i) \in \startset \}$.
The following lemma holds. 
\begin{lemma}\label{lem:same_lcp_values_count}
$\LF_{\mathsf{dist}(i)}(i) \neq \LF_{\mathsf{dist}(j)}(j)$ holds for two distinct integers $i, j \in \mathcal{P}_{t}$, 
where $t \geq 0$ is an integer. 
\end{lemma}
\begin{proof}
We show that $\LCP[i] = \LCP[\LF_{\mathsf{dist}(i)}(i)] - \mathsf{dist}(i)$ holds for an integer $i \in \{1, 2, \ldots, n \}$. 
Let $t$ be an integer in $[0, \mathsf{dist}(i)-1]$. 
Since $\LF_{t}(i) \not \in \startset$, 
$\LCP[\LF_{t}(i)] = \LCP[\LF_{t+1}(i)] - 1$ holds by LF formula. 
The LF formula produces 
$\mathsf{dist}(i)$ equations 
$\LCP[\LF_{0}(i)] = \LCP[\LF_{1}(i)] - 1$, $\LCP[\LF_{1}(i)] = \LCP[\LF_{2}(i)] - 1$, $\ldots$, 
$\LCP[\LF_{\mathsf{dist}(i)-1}(i)] = \LCP[\LF_{\mathsf{dist}(i)}(i)]$ - 1. 
Hence, $\LCP[i] = \LCP[\LF_{\mathsf{dist}(i)}(i)] - \mathsf{dist}(i)$ holds by the $\mathsf{dist}(i)$ equations. 
 
Next, we prove Lemma~\ref{lem:same_lcp_values_count}. 
The two integers $i$ and $j$ must be the same if $\LF_{\mathsf{dist}(i)}(i) = \LF_{\mathsf{dist}(j)}(j)$ 
because $\mathsf{dist}(i) = \mathsf{dist}(j)$ holds by three equations $\LCP[i] = \LCP[j]$, $\LCP[i] = \LCP[\LF_{\mathsf{dist}(i)}(i)] - \mathsf{dist}(i)$, 
and $\LCP[j] = \LCP[\LF_{\mathsf{dist}(j)}(j)] - \mathsf{dist}(j)$. 
The equation $i = j$ contradicts the fact that $i \neq j$. 
Hence, $\LF_{\mathsf{dist}(i)}(i) \neq \LF_{\mathsf{dist}(j)}(j)$ holds. 
\end{proof}
The function $\LF_{\mathsf{dist}(i)}(i)$ maps the integers in $\mathcal{P}_{t}$ into distinct integers in set $\startset$ 
by Lemma~\ref{lem:same_lcp_values_count}. The mapping indicates that  $|\mathcal{P}_{t}| \leq |\startset| = r$ holds for any integer $t$. 
In Figure~\ref{fig:analysis}, the right figure represents the mapping between $\mathcal{P}_{1} = \{ 3, 4 \}$ and $\startset = \{ 1, 2, 5, 6, 7 \}$ on a BWT. 
In this example, $\LF_{1}(3) = 9 \not \in \startset$, $\LF_{2}(3) = 5 \in \startset$, $\LF_{1}(4) = 10 \not \in \startset $, and $\LF_{2}(4) = 6 \in \startset$. 
Hence, $\LF_{\mathsf{dist}(i)}(i)$ maps the two positions $3$ and $4$ in $\mathcal{P}_{1}$ into the two positions $5$ and $6$ in $\startset$, respectively, 
which indicates that $|\mathcal{P}_{1}| \leq |\startset|$ holds.

Finally, we obtain $\max \{ |\mathcal{K}_{0}|, |\mathcal{K}_{1}|, \ldots, |\mathcal{K}_{n}| \} \leq 2r$ by 
$|\mathcal{K}_{t}| \leq 2|\mathcal{P}_{t}|$ and $|\mathcal{P}_{t}| \leq r$. 
Hence, the working space of our traversal algorithm is $O( (r + \sigma)w )$ bits. 
We obtain the following theorem using $\sigma \leq r$. 
\begin{theorem}\label{theo:enumeration_lcp_intervals}
We can output $n$ sets $\{ \repr(P) \mid P \in \mathcal{L}_{0} \}$, $\{ \repr(P) \mid P \in \mathcal{L}_{1} \}$, $\ldots$, 
$\{ \repr(P) \mid P \in \mathcal{L}_{n-1} \}$ in left-to-right order 
in $O(n \log \log_{w} (n/r))$ time and $O(rw)$ bits of working space by processing the RLBWT of $T$. 
\end{theorem}

%% file: 4_applications.tex
\section{Enumeration of characteristic substrings in \texorpdfstring{$O(rw)$}{O(rw)} bits of space}\label{sec:applications}
In this section, we present r-enum, which enumerates maximal repeats, minimal unique substrings, and minimal absent words using RLBWT. 
While the enumeration algorithm proposed by Belazzougui and Cunial~\cite{DBLP:conf/spire/BelazzouguiC15} finds nodes corresponding to characteristic substrings while traversing the Weiner-link tree of $T$, 
r-enum uses our breadth-first traversal algorithm presented in Section~\ref{sec:traverse} instead of their traversal algorithm. 
The next theorem holds by assuming $\sigma \leq r$. 

\begin{theorem}\label{theo:main_result}
R-enum can enumerate (i) maximal repeats, (ii) minimal unique substrings, and (iii) minimal absent words for $T$ 
in $O(n \log \log_{w} (n/r))$, $O(n \log \log_{w} (n/r))$, and $O(n \log \log_{w} (n/r) + occ)$ time, respectively, by processing the RLBWT of $T$ with $O(rw)$ bits of working space, 
where $occ$ is the number of minimal absent words for $T$, and $occ = O(\sigma n)$ holds~\cite{DBLP:journals/ipl/CrochemoreMR98}. 
Here, r-enum outputs rich representation $\repr(P)$, pair $(\stinterval(P'), |P'|)$, and 3-tuple $(\stinterval(P''), |P''|, c)$ 
for a maximal repeat $P$, minimal unique substring $P'$, and minimal absent word $P''c$, respectively, 
where $P, P', P''$ are substrings of $T$, and $c$ is a character. 
\end{theorem}

We prove Theorem~\ref{theo:main_result}(i) for maximal repeats. 
R-enum leverages the following relation between the maximal repeats and nodes in the Weiner-link tree. 
\begin{corollary}\label{cor:maximal_repeat}
A substring $P$ of $T$ is a maximal repeat if and only if 
$P$ satisfies two conditions: 
(i) node $u_{P}$ is explicit~(i.e., $P \in \mathcal{L}_{|P|}$), 
and (ii) $\rank(\startset, b) \neq \rank(\startset, e)$, 
where $\stinterval(P) = [b, e]$. 
\end{corollary}
\begin{proof}
$|\Occ(T, P)| \geq 2$ and $|\Occ(T, P)| > |\Occ(T, Pc)|$ hold for any character $c \in \Sigma$ if and only if node $u_{P}$ is explicit. 
From the definition of BWT, 
$|\Occ(T, P)| > |\Occ(T, cP)|$ holds for any character $c \in \Sigma$ 
if and only if $L[b..e]$ contains at least two distinct characters, i.e., $\rank(\startset, b) \neq \rank(\startset, e)$ holds. 
Hence, Corollary~\ref{cor:maximal_repeat} holds. 
\end{proof}

R-enum traverses the Weiner-link tree of $T$, 
and it verifies whether each explicit node $u_{P}$ represents a maximal repeat by using Corollary~\ref{cor:maximal_repeat}, i.e., 
it verifies $\rank(\startset, b) \neq \rank(\startset, e)$ or $\rank(\startset, b) = \rank(\startset, e)$ by using the two rank queries on $\startset$. 
We output its rich representation $\repr(P)$ if $P$ is a maximal repeat. 
The rich representation $\repr(P)$ for $u_{P}$ stores $[b,e]$, 
and our breadth-first traversal algorithm stores the data structure $R_{\rank}(\startset)$ for rank query on $\startset$. 
Hence, we obtain Theorem~\ref{theo:main_result}(i). 

Similarly, r-enum can also find the nodes corresponding to minimal unique substrings and minimal absent words 
by using their properties while traversing the Weiner-link tree. 
See Appendixes B and C for the proofs of Theorem~\ref{theo:main_result}(ii) and (iii), respectively.

%% file: 5_option.tex
\section{Modified enumeration algorithm for original characteristic substrings and their occurrences}\label{sec:option}
Let $\outputrepr(P)$ be the element representing a characteristic substring $P$ outputted by r-enum. 
In this section, we slightly modify r-enum and provide three additional data structures, $R_{\mathsf{str}}$, $R_{\mathsf{occ}}$, and $R_{\eRD}$, 
to recover the original string $P$ and its occurrences in $T$ from the element $\outputrepr(P)$. 
The three data structures $R_{\mathsf{str}}$, $R_{\mathsf{occ}}$ and $R_{\eRD}$ require $O(rw)$ bits of space 
and support \emph{extract, extract-sa}, and \emph{extended range distinct queries}, respectively.
An extract query returns string $P$ for a given pair $(\stinterval(P), |P|)$. 
An extract-sa query returns all the occurrences of $P$ in $T$~(i.e., $\SA[b..e]$) for a given pair $(\stinterval(P), \SA[b])$, 
where $\stinterval(P) = [b,e]$. 
An extended range distinct query $\eRD(L, b, e, \SA[b])$ returns 4-tuple $(c, p_{c}, q_{c}, \SA[p_{c}])$ for 
each output $(c, p_{c}, q_{c}) \in \RD(L, b, e)$~(i.e., $\eRD(L, b, e, \SA[b]) = \{ (c, p_{c}, q_{c}, \SA[p_{c}]) \mid (c, p_{c}, q_{c}) \in \RD(L, b, e) \}$). 
We omit the detailed description of the three data structures because 
each data structure supports the queries by using the well-known properties of RLBWT. 
Formally, let $k = |\eRD(L, b, e, \SA[b])|$. 
Then, the following lemma holds. 
\begin{lemma}\label{lem:additional_data_structures}
The three data structures $R_{\mathsf{str}}$, $R_{\mathsf{occ}}$, and $R_{\eRD}$ of $O(rw)$ bits of space can support 
extract, extract-sa, and extended range distinct queries in $O(|P| \log \log_{w} (n/r))$, $O((e-b+1) \log \log_{w} (n/r))$, and $O((k+1) \log \log_{w} (n/r))$ time, respectively. 
We can construct the three data structures in $O(n \log \log_{w} (n/r))$ time and $O(rw)$ bits of working space by processing the RLBWT of $T$. 
\end{lemma}
\begin{proof}
See Appendix D. 
\end{proof}

We modify r-enum as follows. 
The modified r-enum outputs pair $(\repr(P), \SA[b])$, 3-tuple $(\stinterval(P'), |P'|, \SA[b'])$, and 4-tuple $(\stinterval(P''), |P''|, c, \SA[b''])$ 
instead of rich representation $\repr(P)$, pair $(\stinterval(P'), |P'|)$, and 3-tuple $(\stinterval(P''), |P''|, c)$, respectively. 
Here, (i) $P, P', P''c$ are a maximal repeat, minimal unique substring, and minimal absent word, respectively, 
and (ii) $b, b', b''$ are the left boundaries of $\stinterval(P)$, $\stinterval(P')$, and $\stinterval(P''c)$, respectively. 
We replace each range distinct query $\RD(L, b, e)$ used by r-enum with the corresponding extended range distinct query 
$\eRD(L, b, e, \SA[b])$ to compute the sa-values $\SA[b]$, $\SA[b']$, and $\SA[b'']$. 
See Appendix E for a detailed description of the modified r-enum. 
Formally, the following theorem holds. 
\begin{theorem}\label{theo:main_result2}
R-enum can also output the sa-values $\SA[b]$, $\SA[b']$, and $\SA[b'']$ for each maximal repeat $P$, minimal unique substring $P'$, 
and minimal absent word $P''c$, respectively, by using an additional data structure of $O(rw)$ bits of space. 
Here, $b, b', b''$ are the left boundaries of $\stinterval(P)$, $\stinterval(P')$, and $\stinterval(P''c)$, respectively, 
and the modification does not increase the running time. 
\end{theorem}
\begin{proof}
See Appendix E. 
\end{proof}

We compute each characteristic substring and all the occurrences of the substring in $T$ by applying 
extract and extract-sa queries to the corresponding element outputted by the modified r-enum. 
For example, the outputted pair $(\repr(P), \SA[b])$ for a maximal repeat $P$ contains $\stinterval(P), |P|$, and $\SA[b]$. 
Hence, we can obtain $P$ by applying an extract query to pair $(\stinterval(P), |P|)$. 
Similarly, we can obtain all the occurrences of $P$ in $T$ by applying an extract-sa query to pair $(\stinterval(P), \SA[b])$. 
Note that we do not need to compute the occurrences of minimal absent words in $T$ 
since the words do not occur in $T$. 

%% file: 6_experiments.tex
\section{Experiments}\label{sec:exp}
\begin{table}[t]
    \footnotesize
    \caption{
    Details of dataset. 
    Table details data size in megabytes (MB), string length~($n$), 
    alphabet size~($\sigma$), number of runs in BWT~($r$), 
    and number of maximal repeats~($m$) for each piece of data.
    }
    \label{table:fileinfo} 

\begin{tabular}[t]{l||r|r|r|r|r}
Data name & Data size [MB] & $n$ &$\sigma$ & $r$ & $m$ \\ \hline
einstein.de.txt & 93  & 92,758,441  & 118 & 101,370  & 79,469  \\
einstein.en.txt & 468  & 467,626,544  & 140 & 290,239  & 352,590  \\ 
world leaders & 47  & 46,968,181  & 90 & 573,487  & 521,255  \\ 
influenza & 155  & 154,808,555  & 16 & 3,022,822  & 7,734,058  \\ 
kernel & 258  & 257,961,616  & 161 & 2,791,368  & 1,786,102  \\ 
cere & 461  & 461,286,644  & 6 & 11,574,641  & 10,006,383  \\ 
coreutils & 205  & 205,281,778  & 237 & 4,684,460  & 2,963,022  \\ 
Escherichia Coli & 113  & 112,689,515  & 16 & 15,044,487  & 12,011,071  \\ 
para & 429  & 429,265,758  & 6 & 15,636,740  & 13,067,128  \\ 
\hline
100genome & 307,705  & 307,705,110,945  & 6 & 36,274,924,494  & 52,172,752,566  \\ 
\end{tabular}    

\end{table}

We demonstrate the effectiveness of our r-enum for enumerating maximal repeats on a benchmark dataset of highly repetitive strings 
in a comparison with the state-of-the-art enumeration algorithms 
of the OT~\cite{DBLP:conf/sdm/OkanoharaT09}, BBO~\cite{DBLP:conf/spire/BellerBO12} and BC~\cite{DBLP:conf/spire/BelazzouguiC15} methods, 
which are reviewed in Section~\ref{sec:intro}.

The OT method enumerates maximal repeats by using the BWT and enhanced suffix array of a given string, where  
the enhanced suffix array consists of suffix and LCP arrays.
The OT method runs in $O(n)$ time and with $O(n \log n)$ bits of working space for $T$. 

The BBO method traverses the Weiner-link tree of a given string $T$ by using Lemma~\ref{lem:weiner_algorithm} and a breadth-first search, and 
it outputs maximal repeats by processing all the nodes in the tree. 
We used the SDSL library~\cite{gbmp2014sea} for an implementation of the Huffman-based wavelet tree to support range distinct queries, 
and we did not implement 
the technique of Beller et al. for storing a queue for suffix-tree intervals in $n + o(n)$ bits. 
Hence, our implementation of the BBO method takes the BWT of $T$ and runs in $O(n \log \sigma)$ time and with $|WT_{\mathbf{huff}}| + O(\max \{ |\mathcal{K}_{0}|, |\mathcal{K}_{1}|, \ldots, |\mathcal{K}_{n}| \} )w)$ bits of working space, where 
(i) $|WT_{\mathbf{huff}}| = O(n \log \sigma)$ is the size of the Huffman-based wavelet tree, and 
(ii) $\mathcal{K}_{0}, \mathcal{K}_{1}, \ldots, \mathcal{K}_{n}$ are introduced in Section~\ref{sec:analysis}. 
The latter term can be bounded by $O(rw)$ bits by 
applying the analysis described in Section~\ref{sec:analysis} to their enumeration algorithm. 

The BC method also traverses the Weiner-link tree by using Lemma~\ref{lem:weiner_algorithm} and a depth-first search. 
The method stores a data structure for range distinct queries and a stack data structure of size $O(\sigma^{2} \log^2 n)$ bits for the depth-first search.
We also used the SDSL library~\cite{gbmp2014sea} for the Huffman-based wavelet tree to support range distinct queries.  
Hence, our implementation of the BC method runs in 
$O(n \log \sigma)$ time and with $|WT_{\mathbf{huff}}| + O(\sigma^{2} \log^{2} n)$ bits of working space. 

We used a benchmark dataset of nine highly repetitive strings in the Pizza \& Chili corpus downloadable from \url{http://pizzachili.dcc.uchile.cl}. 
In addition, we demonstrated the scalability of r-enum by enumerating maximal repeats on 
a huge string of 100 human genomes with 307 gigabytes built from 1,000 human genomes~\cite{1000Genomes}. 
Table~\ref{table:fileinfo} details our dataset. 

We used memory consumption and execution time as evaluation measures for each method. 
Since each method takes the BWT of a string as an input and outputs maximal repeats, 
the execution time consists of two parts: 
(i) the preprocessing time for constructing data structures built from an input BWT, and 
(ii) the enumeration time after the data structures are constructed.
We performed all the experiments on 48-core Intel Xeon Gold 6126 (2.60 GHz) CPU with 2 TB of memory. 


\begin{table}[t]
    \footnotesize
    \caption{Peak memory consumption of each method in mega bytes (MB). 
    }
    \label{table:memory} 
    \center{	
\begin{tabular}[t]{|l||r||r|r|r|r|}
\hline
 &  & \multicolumn{4}{|c|}{Memory (MB)}  \\ \hline
Data name & Data size [MB] & r-enum & BBO & BC & OT \\ \hline\hline
einstein.de.txt & 93  & {\bf 2}  & 100  & 100  & 1,642  \\ 
einstein.en.txt & 468  & {\bf 4}  & 488  & 488  & 8,278  \\ 
world leaders & 47  & {\bf 5}  & 37  & 37  & 832  \\ 
influenza & 155  & {\bf 18}  & 77  & 73  & 2,741  \\ 
kernel & 258  & {\bf 21}  & 297  & 292  & 4,567  \\ 
cere & 461  & {\bf 86}  & 265  & 224  & 8,166  \\ 
coreutils & 205  & {\bf 32}  & 268  & 237  & 3,634  \\ 
Escherichia Coli & 113  & 109  & 116  & {\bf 55}  & 1,995  \\ 
para & 429  & {\bf 116}  & 262  & 204  & 7,599  \\ \hline
\end{tabular}    
}
\end{table}
\begin{table}[t]
    \footnotesize
    \caption{Execution time of each method in seconds (s).}
    \label{table:time} 
    \center{	
\begin{tabular}[t]{|l||r|r|r|r|}
\hline
& \multicolumn{4}{|c|}{Execution time (s)}   \\ \hline
Data name & r-enum & BBO & BC & OT \\ \hline\hline
einstein.de.txt & 172  & 84  & 70  & {\bf 13}   \\
einstein.en.txt & 856  & 487  & 387  & {\bf 80}   \\
world leaders & 97  & 24  & 16  & {\bf 10}   \\ 
influenza & 267  & 69  & 49  & {\bf 30}   \\
kernel & 559  & 281  & 178  & {\bf 57}   \\
cere & 985  & 277  & 186  & {\bf 110}   \\
coreutils & 445  & 285  & 179  & {\bf 42}   \\ 
Escherichia Coli & 253  & 68  & 40  & {\bf 29}   \\ 
para & 961  & 262  & 173  & {\bf 102} \\ \hline
\end{tabular}
}
\end{table}

\subsection{Experimental results on benchmark dataset}

In the experiments on the nine highly repetitive benchmark strings, we ran each method and with a single thread.
Table~\ref{table:memory} shows the peak memory consumption of each method. 
The BBO and BC methods consumed approximately 1.0-2.0 and 0.9-1.5 bytes per byte of input, respectively.
The memory usage of the BC method was no more than that of the BBO method on each of the nine strings. 
The OT method consumed approximately 18 bytes per character, 
which was larger than the memory usage of the BBO method. 
Our r-enum consumed approximately 7-23 bytes per run in BWT. 
The memory usage of r-enum was the smallest on most of the nine strings except for the file Escherichia Coli.
In the best case, the memory usage of r-enum was approximately 127 times less than that of the BC method on einstein.en.txt 
because the ratio $n/r \approx 1611$ and alphabet size $\sigma = 140$ were large.

Table~\ref{table:time} shows the execution time for each method. 
The OT method was the fastest among all the methods, and it took approximately 13-110 seconds.
The execution times of the BBO and BC methods were competitive, 
and each execution of them was finished within 487 seconds on the nine strings.
R-enum was finished in 985 seconds even for the string data (cere) with the longest enumeration time. 
These results show that r-enum can space-efficiently enumerate maximal substrings in a practical amount of time.

\begin{table}[t]
    \footnotesize
    \caption{Execution time in hours and peak memory consumption in mega-bytes (MB) of r-enum on 100genome.}
    \label{table:t100genome} 
    \center{	
\begin{tabular}[t]{r|r|r|r}
Data size [MB] & $n/r$ & Execution time (hours) & Memory (MB) \\ \hline
307,705  & 8.5 & 25 & 319,949  \\ 
\end{tabular}
}
\end{table}

\subsection{Experimental results on 100 human genomes}

We tested r-enum on 100genome, which is a 307-gigabyte string of 100 human genomes. 
For this experiment, we implemented computations of Weiner-links from nodes with the same depth in parallel, and 
we ran the parallelized r-enum with 48 threads.


Table~\ref{table:t100genome} shows the total execution time and peak memory consumption of r-enum on 100genome. 
R-enum consumed approximately 25 hours and 319 gigabytes of memory for enumeration. 
The result demonstrates the scalability and practicality of r-enum for enumerating maximal repeats on a huge string existing in the real-world. 
The execution time was 0.6 seconds per $10^6$ characters. 
The memory consumption was almost the same as the size of the input file. 
It was relatively large compared with the memory consumption of r-enum on the nine benchmark strings because 
the ratio of $n/r \approx 8.4$ for 100genome was small.

\section{Conclusion}
We presented r-enum, which can enumerate maximal repeats, minimal unique substrings, and minimal absent words working in $O(rw)$ bits of working space. 
Experiments using a benchmark dataset of highly repetitive strings showed that r-enum is more space-efficient than the previous methods. 
In addition, we demonstrated the applicability of r-enum to a huge string by performing experiments on a 300-gigabyte string of 100 human genomes.
Our future work is to reduce the running time and working space of r-enum. 

We showed that breadth-first traversal of a Weiner-link tree can be performed in $O(rw)$ bits of working space. 
The previous breadth-first traversal algorithm by the BBO method requires $|RD| + O(n)$ bits of working space, 
where $|RD|$ is the size of a data structure supporting range distinct queries on the BWT of $T$. 
In addition to enumerations of characteristic substrings, traversal algorithms of a Weiner-link tree can be used for constructing three data structures:  
the (i) LCP array, (ii) suffix tree topology, and (iii) the merged BWT of two strings~\cite{DBLP:journals/tcs/PrezzaR21}. 
Constructing these data structures in $O(rw)$ bits of working space by modifying r-enum would be an interesting future work.

We also showed that the data structure $X$ for our traversal algorithm supported two operations: 
(i) a range distinct query in $O(\log \log_{w} (n/r))$ time per output element and 
(ii) computations of LF function in $O(\log \log_{w} (n/r))$ time. 
The result can replace the pair $(|RD|, d)$ presented in Table~\ref{table:result} with pair $(O(rw), O(\log \log_{w} (n/r)))$, which improves
the BBO and BC methods so that they work in $O(rw + n)$ and $O(rw + \sigma^{2} \log^{2} n)$ bits of working space, respectively.  
We think that the working space of the BC method with $X$ is practically smaller than that of r-enum, 
because $\sigma$ is practically smaller than $r$ in many cases. 
This insight indicates that we can enumerate characteristic substrings with a lower memory consumption of $(rw)$ bits
by combining r-enum with the BC method even for a large alphabet size.
Thus, the following method could improve the space efficiency for enumerations:
executing the BC method with $X$ for a small alphabet (i.e., $\sigma$ < $\sqrt{r / \log n}$) and executing r-enum for a large alphabet.

%% file: appendix.tex
\section*{Appendix A: Algorithm for constructing data structures in Section~\ref{sec:data_structures}}
Recall that our traversal algorithm uses six data structures: 
(i) the RLBWT of $T$, 
(ii) $\repr(\varepsilon)$, 
(iii) $D_{\LF}$, 
(iv) $R_{\rank}(\startset)$, 
(v) $R_{\RD}(L')$, 
and (vi) an empty array $X$ of size $\sigma$. 

Let $\delta$ be the permutation of $[1,r]$ satisfying either of two conditions for two distinct integers $i, j \in \{1, 2, \ldots, r \}$: 
(i) $L[\ell(\delta[i])] \prec L[\ell(\delta[j])]$ or (ii) $L[\ell(\delta[i])] = L[\ell(\delta[j])]$ and $i < j$. 
Then $D_{\LF}[1] = 1$ and $D_{\LF}[i] = D_{\LF}[i-1] + |L[\ell(\delta[i-1])..\ell(\delta[i-1] + 1) - 1]|$ hold by LF formula. 
We construct the permutation $\delta$ in $O(n)$ time and $O(rw)$ bits of working space by using LSD radix sort. 
After that we construct the array $D_{\LF}$ in $O(r)$ time by using $\delta$. 

Next, 
we construct set $\startset$ and string $L'$ in $O(r)$ time by processing the RLBWT of $T$. 
After that, we construct $R_{\rank}(\startset)$ in 
$O(|\startset| \log \log_{w} (n/r))$ time and $O(|\startset| w)$ bits of working space by processing $\startset$~\cite[Appendix A.1]{10.1145/3375890}. 
Similarly, we construct $R_{\RD}(L')$ in $O(r)$ time and $O(r \log \sigma)$ bits of working space~\cite[Lemma 3.17]{DBLP:journals/talg/BelazzouguiCKM20}. 

Finally, $\repr(\varepsilon)$ consists of 3-tuple $([1,n]$, $\RD(L, 1, n)$, $0)$. 
We compute $\RD(L, 1, n)$ by using Lemma~\ref{lem:range_query}. 
Hence the construction time is $O(n \log \log_{w} (n/r))$ in total, and the working space is $O(rw)$ bits.

\section*{Appendix B: Proof of Theorem~\ref{theo:main_result}(ii)}
For simplicity, we focus on minimal unique substrings with a length of at least $2$.
Every minimal unique substring with a length of at least $2$ is a substring $cPc'$ of $T$,
where $c, c'$ are characters, and $P$ is a string with a length of at least 0. 
R-enum uses an array $\rightarr_{P}$ of size $\sigma$ for detecting a minimal unique substring $cPc'$. 
$\rightarr_{P}[c]$ stores $|\Occ(T, Pc)|$ for each $c \in \Sigma$. 
A substring $cPc'$ of $T$ is a minimal unique substring if and only if $cPc'$ satisfies three conditions: 
(i) $|\Occ(T, cPc')| = 1$, (ii) $|\Occ(T, cP)| \geq 2$, and (iii) $\rightarr_{P}[c] \geq 2$ hold. 
We can verify the three conditions by using $\repr(P)$ and Lemma~\ref{lem:weiner_algorithm}, 
and hence, the following lemma holds. 

\begin{lemma}\label{lem:mus}
Let $\mathcal{M}(P)$ be the set of minimal unique substrings such that the form of each minimal unique substring is $cPc'$, 
where $c, c'$ are characters, and $P$ is a given string. 
We can compute the output by r-enum for the set $\mathcal{M}(P)$~(i.e., $\{ (\stinterval(cPc'), |cPc'|) \mid cPc' \in \mathcal{M}(P) \}$) 
by using (i) $\repr(R)$, (ii) the data structures presented in Section~\ref{sec:data_structures}, 
and (iii) an empty array $X'$ of size $\sigma$. 
The running time and working space are $O(h \log \log_{w} (n/r))$ and $O((\sigma + h')w)$ bits, respectively. 
Here, $h = \sum_{Pc' \in \children(P)} |\Weiner(Pc')|$, and 
$h' = \sum_{cP \in \Weiner(P) \cap \mathcal{L}_{|P|+1}} |\children(cP)|$. 
\end{lemma}
\begin{proof}
$\rightarr_{P}[c]$ is stored in $\repr(P)$ for $c \in \Sigma$,  
and the pair $(\stinterval(cPc'), |cPc'|)$ is stored in $\repr(cP)$. 
We compute the output by four steps.  
(i) We convert $X'$ into $\rightarr_{P}$ by processing $\repr(P)$. 
(ii) We compute the rich representations for all the strings in $\Weiner(P)$ by Lemma~\ref{lem:weiner_algorithm}.  
(iii) We process the rich representation for each string $cP \in \Weiner(P)$ 
and output pair $(\stinterval(cPc'), |cPc'|)$ for each child $cPc'$ if $|\stinterval(cPc')| = 1$ and $\rightarr_{P}[c'] \geq 2$. 
(iv) We recover $X'$ from $\rightarr_{P}$. 
Hence, Lemma~\ref{lem:mus} holds. 
\end{proof}

Node $u_{P}$ is always explicit~(i.e., $u_{P}$ is a node of the Weiner-link tree) if a substring $cPc'$ is a minimal unique substring of $T$ 
because $|\Occ(T, P)| > |\Occ(T, Pc')|$ holds from the definition of the minimal unique substring. 
Hence, we can compute $(\stinterval(cPc'), |cPc'|)$ for each minimal unique substring $cPc'$ 
in $T$ by applying Lemma~\ref{lem:mus} to all the nodes in the Weiner-link tree. 

R-enum prepares an empty array $X'$ of size $\sigma$ and computes the output for set $\mathcal{M}(P)$ by applying Lemma~\ref{lem:mus} to each node $u_{P}$. 
The running time and working space are $O((\sum_{t = 0}^{n} H_{t}) \log \log_{w} (n/r) ) = O(n \log \log_{w} (n/r))$ and 
$O( (r + \sigma + \max \{ |\mathcal{K}_{0}|, |\mathcal{K}_{1}|, \ldots, |\mathcal{K}_{n}| \} )w )) = O(rw)$ bits, respectively. 
Here, $H_{t}$ and $\mathcal{K}_{t}$ are the terms introduced in Section~\ref{sec:analysis}. 

\section*{Appendix C: Proof of Theorem~\ref{theo:main_result}(iii)}
For simplicity, we focus on minimal absent words with a length of at least $2$. 
Minimal absent words have similar properties to the properties of the minimal unique substrings explained in the proof of Theorem~\ref{theo:main_result}(ii), 
i.e., the characteristic substrings have two properties: 
(i)
a string $cPc'$ is a minimal absent word for $T$ if and only if 
$cPc'$ satisfies three conditions: 
(1) $|\Occ(T, cPc')| = 0$, (2) $|\Occ(T, cP)| \geq 1$, and (3) $\rightarr_{P}[c'] \geq 1$ hold, and  
(ii) node $u_{P}$ is always explicit if a substring $cPc'$ is a minimal absent word for $T$ 
because $|\Occ(T, P)| > |\Occ(T, Pc')|$ holds from the definition of the minimal absent word. 
We obtain the following lemma by modifying the proof of Lemma~\ref{lem:mus}. 
\begin{lemma}\label{lem:maw}
Let $\mathcal{W}(P)$ be the set of minimal absent words such that the form of each minimal unique substring is $cPc'$, 
where $c, c'$ are characters, and $P$ is a given string. 
We can compute the output by r-enum for the set $\mathcal{W}(P)$~(i.e., $\{ (\stinterval(cP), |cP|, c') \mid cPc' \in \mathcal{W}(P) \}$) 
by using (i) $\repr(R)$, (ii) the data structures presented in Section~\ref{sec:data_structures}, 
and (iii) an empty array $X'$ of size $\sigma$. 
The running time and working space are $O(h \log \log_{w} (n/r) + |\mathcal{W}(P)| )$ and $O((\sigma + h')w)$ bits, respectively. 
\end{lemma}
By using Lemma~\ref{lem:maw}, 
we can compute $(\stinterval(cP), |cP|, c')$ for each minimal absent word $cPc'$ for $T$ by applying Lemma~\ref{lem:maw} to all the nodes in the Weiner-link tree. 

R-enum prepares an empty array $X'$ of size $\sigma$ and computes the output for set $\mathcal{W}(P)$ by applying Lemma~\ref{lem:mus} to each node $u_{P}$. 
The running time and working space are $O(n \log \log_{w} (n/r) + occ)$ and $O(rw)$ bits, respectively.

\section*{Appendix D: Proof of Lemma~\ref{lem:additional_data_structures} }
\begin{proof}[Proof for data structure $R_{\mathsf{str}}$]
Data structure $R_{\mathsf{str}}$ consists of 
(i) the RLBWT of $T$, (ii) $\pi$, (iii) $D_{\LF}$, (iv) $R_{\rank}(\startset')$, and (v) $R_{\rank}(\startset)$.
Here, (i) $\pi$ is the permutation on $\{1, 2, \ldots, r \}$ satisfying $\LF(\ell(\pi[1])) < \LF(\ell(\pi[2])) < \cdots < \LF(\ell(\pi[r]))$, 
(ii) $D_{\LF}$ is the array introduced in Section~\ref{sec:traverse}, and
(iii) $\startset' = \{ \LF(\ell(\pi[1]))$ , $\LF(\ell(\pi[2]))$, $\ldots$, $\LF(\ell(\pi[r])) \}$. 
$\pi[i] = \delta[i]$ holds for any integer $i \in \{1, 2, \ldots, r \}$. 
$\startset'$ and $D_{\LF}$ can be constructed in $O(r)$ time after the permutation $\pi$ is constructed. 
We already showed that $\delta$ could be constructed in $O(n)$ time by processing the RLBWT of $T$, which was explained in Appendix A. 
Hence, we can construct $R_{\mathsf{str}}$ in $O(n \log \log_{w} (n/r))$ time and $O(rw)$ bits of working space.

We introduce the inverse function $\LF^{-1}$ of LF function~(i.e., $\LF^{-1}(\LF(i)) = i$ holds for $i \in \{ 1, 2, \ldots, n \}$) to solve the extract query.  
Recall that $\LF(i) = D_{\LF}[x] + (i - \ell(x))$ holds, which is shown in Section~\ref{sec:data_structures}, 
where $x = \rank(\startset, i)$. 
Similarly, $\LF^{-1}(i) = \ell(\pi[y]) + (i - D_{\LF}[\pi[y]])$ holds by the LF formula for any integer $i \in \{ 1, 2, \ldots, n \}$, 
where $y = \rank(\startset', i)$. 
Hence, we can compute $\LF^{-1}(i)$ in $O(\log \log_{w} (n/r))$ time by using the data structure $R_{\mathsf{str}}$. 

Let $\LF^{-1}_{x}$ be the function that returns the position obtained by recursively applying the inverse LF function to $i$ $x$ times~(i.e., 
$\LF^{-1}_{0}(i) = i$, and $\LF^{-1}_{x}(i) = \LF^{-1}(\LF^{-1}_{x-1}(i))$). 
Then, $T[\SA[i]..\SA[i] + d -1] = L[\LF^{-1}_{1}(i)], L[\LF^{-1}_{2}(i)], \ldots, L[\LF^{-1}_{d}(i)]$ holds for any integer $d \geq 1$ 
because $\SA[\LF^{-1}(i)] = \SA[i] + 1$ holds unless $\SA[i] = n$. 
$R_{\mathsf{str}}$ can support random access to the BWT $L$ in $O(\log \log_{w} (n/r))$ time using a rank query on set $\startset$. 
Hence, we can compute $T[\SA[i]..\SA[i] + d -1]$ in $O(d \log \log_{w} (n/r))$ time by using $R_{\mathsf{str}}$ for two given integers $i$ and $d$. 

We explain an algorithm solving an extract query for a given rich representation $\repr(P)$. 
Let $\stinterval(P) = [b, e]$. 
Then, $\SA[b]$ stores the index of a suffix of $T$ having $P$ as a prefix, i.e., 
$T[\SA[b]..\SA[b] + |P| -1] = P$ holds. 
We recover the prefix $P$ from $\SA[b]$ by using $R_{\mathsf{str}}$. 
Hence $R_{\mathsf{str}}$ supports the extract query in $O(|P| \log \log_{w} (n/r))$ time. 

\end{proof}
\begin{proof}[Proof for data structure $R_{\mathsf{occ}}$]
Next, we leverage a function $\phi$ to solve the extract-sa query. 
The function $\phi(\SA[i])$ returns $\SA[i+1]$ for $i \in \{ 1, 2, \ldots, n-1 \}$. 
$R_{\phi}$ is a data structure of $O(rw)$ bits proposed by Gagie et al.~\cite{10.1145/3375890}, 
and we can compute $\phi$ function in $O(\log \log_{w} (n/r))$ time by using $R_{\phi}$. 
The data structure can be constructed in $O(n \log \log_{w} (n/r))$ time and $O(rw)$ bits of working space by processing the RLBWT of $T$~\cite{10.1145/3375890}. 
The second data structure $R_{\mathsf{occ}}$ consists of $R_{\phi}$, 
and we solve the extract-sa query by recursively applying the function $\phi$ to $\SA[b]$ $(e-b)$ times. 
Hence $R_{\mathsf{occ}}$ can support the extract-sa query in $O((e-b+1) \log \log_{w} (n/r))$ time.
\end{proof}

\begin{proof}[Proof for data structure $R_{\eRD}$]
Let $D_{\SA}$ be an array of size $r$ such that 
$D_{\SA}[i]$ stores the sa-value at the starting position of the $i$-th run in BWT $L$ for $i \in \{1, 2, \ldots, r \}$, 
i.e., $D_{\SA}[i] = \SA[\ell(i)]$. 
Let $(c, p_{c}, q_{c}, \SA[p_{c}])$ be a 4-tuple outputted by query $\eRD(L, b, e, \SA[b])$ 
and $x$ be the index of the run containing character $L[p_{c}]$~(i.e., $x = \rank(\startset, p_{c})$). 
Then $\SA[p_{c}] = D_{\SA}[x]$ if $p_{c} = \ell(x)$; otherwise $\SA[p_{c}] = \SA[b]$ holds  
because $p_{c}$ is equal to $\ell(x)$ or $b$ under Lemma~\ref{lem:range_query}. 
The relationship among $\SA[p_{c}]$, $D_{\SA}[x]$, and $\SA[b]$ is called the \emph{toehold lemma}~(e.g., \cite{DBLP:journals/algorithmica/PolicritiP18,10.1145/3375890}). 
The toehold lemma indicates that 
we can compute $\SA[p_{c}]$ in $O(\log \log_{w} (n/r))$ time for each output $(c, p_{c}, q_{c}) \in \RD(L, b, e)$ by using 
(i) the array $D_{\SA}$, (ii) data structure $R_{\rank}(\startset)$, and (iii) the RLBWT of $T$ 
if we know the first sa-value $\SA[b]$ in $[b,e]$. 

Next, we explain $R_{\eRD}$. 
$R_{\eRD}$ consists of $R_{\RD}(L')$, $D_{\SA}$, $R_{\rank}(\startset)$, and the RLBWT of $T$. 
We already showed that we could construct $R_{\RD}(L')$ and $R_{\rank}(\startset)$ in $O(n \log \log_{w} (n/r))$ time by processing the RLBWT of $T$. 
We construct the array $D_{\SA}$ by computing all the sa-values in $\SA[1..n]$ in left-to-right order by using data structure $R_{\phi}$, 
and $R_{\phi}$ can be constructed in $O(n \log \log_{w} (n/r))$ time by processing the RLBWT of $T$. 
Hence, the construction time for $R_{\eRD}$ is $O(n \log \log_{w} (n/r))$ time in total, 
and the working space is $O(rw)$ bits. 

We solve extended range distinct query $\eRD(L, b, e, \SA[b])$ by using the toehold lemma after 
solving the corresponding range distinct query $\RD(L, b, e)$ by using $R_{\RD}(L')$, $R_{\rank}(\startset)$, and the RLBWT of $T$. 
The running time is $O((k+1) \log \log_{w} (n/r))$, 
where $k = |\eRD(L, b, e, \SA[b])|$. 
\end{proof}

\section*{Appendix E: Proof of Theorem~\ref{theo:main_result2}}
We extend Lemma~\ref{lem:weiner_algorithm}. 
Let $\erepr(P)$ for $P$ be a 4-tuple $(\stinterval(P)$, $\{ (c_{1}, b_{1}, e_{1}, \SA[b_{1}])$, 
$(c_{2}, b_{2}, e_{2}, \SA[b_{2}])$, $\ldots$, $(c_{k}, b_{k}, e_{k}, \SA[b_{k}] ) \}$, $|P|, \SA[b])$.  
Here, 
(i) $b$ is the left boundary of $\stinterval(P)$, 
(ii) $Pc_{1}, Pc_{2}, \ldots, Pc_{k}$ are strings represented by the children of node $u_{P}$, 
and (iii) $[b_{i}, e_{i}] = \stinterval(Pc_{i})$ for $i \in [1, k]$. 
We call $\erepr(P)$ an \emph{extended rich representation}. 

Let $\erepr(cP) = (\stinterval(cP)$, $\{ (c'_{1}, b'_{1}, e'_{1}, \SA[b'_{1}])$, $(c'_{2}, b'_{2}, e'_{2}, \SA[b'_{2}])$, $\ldots$, $(c'_{k'}, b'_{k'}, e'_{k'}, \SA[b'_{k'}] ) \}$, $|cP|, \SA[b'])$ for a character $c$. 
Let $x_{i}$ be an integer such that $\LF(x_{i}) = b'_{i}$ for a 4-tuple $(c'_{i}, b'_{i}, e'_{i}, \SA[b'_{i}])$ in $\erepr(P)$, 
and let $y(i)$ be an integer such that $x_{i} \in [b_{y(i)}, e_{y(i)}]$ holds. 
Then, there exists a tuple $(\hat{c}, p_{\hat{c}}, q_{\hat{c}}, \SA[p_{\hat{c}}]) \in \eRD(L, b_{y(i)}, e_{y(i)}, \SA[b_{y(i)}])$ such that 
$p_{\hat{c}} = x_{i}$ holds. 
$\SA[b'_{i}] = \SA[p_{\hat{c}}] - 1$ holds by LF function. 
Next, let $j$ be an integer such that $b'_{j}$ is the smallest in set $\{ b'_{1}, b'_{2}, \ldots, b'_{k'} \}$. 
Then $\SA[b'] = \SA[b'_{j}]$ holds because $b'_{j}$ is equal to the left boundary of $\stinterval(cP)$. 
Hence, Lemma~\ref{lem:weiner_algorithm} can output extended rich representations instead of rich representations 
by replacing the range distinct queries used by the algorithm of Lemma~\ref{lem:weiner_algorithm} with the corresponding extended range distinct queries. 
Formally, the following lemma holds. 

\begin{lemma}\label{lem:modified_weiner_algorithm}
We can compute set $\{ \erepr(cP) \mid cP \in \Weiner(P) \}$ for a given rich representation $\erepr(P)$ 
in $O(h \log \log_{w} (n/r))$ time by using $R_{\eRD}$ and the six data structures introduced in Section~\ref{sec:data_structures}, 
where $h = \sum_{Pc' \in \children(P)} |\Weiner(Pc')|$. 
\end{lemma}

Next, we modify our traversal algorithm for the Weiner-link tree of $T$. 
The modified traversal algorithm uses Lemma~\ref{lem:modified_weiner_algorithm} instead of Lemma~\ref{lem:weiner_algorithm}. 
Hence, we obtain the following lemma. 
\begin{lemma}\label{lem:modified_traversal}
We can output $n$ sets $\{ \erepr(P) \mid P \in \mathcal{L}_{0} \}$, $\{ \erepr(P) \mid P \in \mathcal{L}_{1} \}$, $\ldots$, 
$\{ \erepr(P) \mid P \in \mathcal{L}_{n-1} \}$ in left-to-right order 
in $O(n \log \log_{w} (n/r))$ time and $O(rw)$ bits of working space by processing the RLBWT of $T$. 
\end{lemma}
Finally, we prove Theorem~\ref{theo:main_result2} by using the modified traversal algorithm, i.e., Lemma~\ref{lem:modified_traversal}.
\begin{proof}[Proof for maximal repeats]
The node $u_{P}$ representing a maximal repeat $P$ is explicit, 
and hence, $\erepr(P)$ is outputted by the modified traversal algorithm. 
R-enum uses the modified traversal algorithm instead of our original traversal algorithm. 
Hence, r-enum can output the extended rich representations for all the maximal repeats in $T$ without increasing the running time. 
\end{proof}
\begin{proof}[Proof for minimal unique substrings]
Let $cPc'$ be a minimal unique substring of $T$ such that its occurrence position is $\SA[b']$. 
Recall that r-enum computes $\repr(cP)$ by applying Lemma~\ref{lem:weiner_algorithm} to $\repr(P)$ 
and outputs $(\stinterval(cPc'), |cPc'|)$ by processing $\repr(cP)$. 
The extended rich representation $\erepr(cP)$, which corresponds to $\repr(cP)$, 
contains $\stinterval(cPc'), |cPc'|$, and $\SA[b']$. 
The modified r-enum (i) traverses the Weiner-link tree by using the modified traversal algorithm, 
(ii) computes $\erepr(cP)$ by applying Lemma~\ref{lem:modified_weiner_algorithm} to $\erepr(P)$, 
and (iii) outputs $(\stinterval(cPc'), |cPc'|, \SA[b'])$ by processing $\erepr(cP)$. 
The running time is $O(n \log \log_{w} (n/r))$ time in total. 
\end{proof}

\begin{proof}[Proof for minimal absent words]
Let $cPc'$ be a minimal absent word for $T$, 
and let $b''$ be the left boundary of $\stinterval(cP)$. 
$\erepr(cP)$ contains the sa-value $\SA[b'']$, 
and hence, we can compute $(\stinterval(cP), |cP|, c', \SA[b'])$ for the minimal absent word $cPc'$ 
by modifying the algorithm used by the modified r-enum for minimal unique substrings. 

The modified r-enum (i) traverses the Weiner-link tree by using the modified traversal algorithm, 
(ii) computes $\erepr(cP)$ by applying Lemma~\ref{lem:modified_weiner_algorithm} to $\erepr(P)$, 
and (iii) outputs $(\stinterval(cP), |cP|, c', \SA[b'])$ by processing $\erepr(cP)$. 
The running time is $O(n \log \log_{w} (n/r) + occ)$ time in total. 
\end{proof}

\section*{Appendix F: Omitted table}

\begin{table}[ht]
    \footnotesize
    \caption{Execution time of each method. Execution time is separately presented as enumeration and preprocessing times in seconds (s).}
    \label{table:time2} 
    \center{	
\begin{tabular}[t]{|l|r|r|r|r||r|r|r|r|}
\hline
& \multicolumn{4}{|c||}{Preprocessing time [s]} & \multicolumn{4}{|c|}{Enumeration time [s]}  \\ \hline
Data name & r-enum & BBO & BC & OT & r-enum & BBO & BC & OT \\ \hline\hline
einstein.de.txt & 1  & 1  & 1  & 7  & 171  & 83  & 69  & {\bf 6}  \\ \hline
einstein.en.txt & 3  & 3  & 3  & 50  & 853  & 484  & 384  & {\bf 30}  \\ \hline
world leaders & 1  & 1  & 1  & 7  & 96  & 23  & 15  & {\bf 3}  \\ \hline
influenza & 2  & 1  & 1  & 17  & 265  & 68  & 48  & {\bf 13}  \\ \hline
kernel & 2  & 2  & 2  & 41  & 557  & 279  & 176  & {\bf 16}  \\ \hline
cere & 5  & 3  & 3  & 78  & 980  & 274  & 183  & {\bf 32}  \\ \hline
coreutils & 2  & 2  & 2  & 29  & 443  & 283  & 177  & {\bf 13}  \\ \hline
Escherichia Coli & 5  & 1  & 1  & 18  & 248  & 67  & 39  & {\bf 11}  \\ \hline
para & 6  & 3  & 3  & 71  & 955  & 259  & 170  & {\bf 31}  \\ \hline
\end{tabular}
}
\end{table}